%% file: ex_article.tex
\documentclass[onefignum,onetabnum]{siamart171218}
\usepackage{amssymb}
\usepackage[sort,compress]{cite}
\usepackage[version=4]{mhchem}  

\usepackage{mathtools}
\usepackage{array}
\usepackage{tabularx}
\usepackage{wrapfig}

\newcolumntype{C}{>{\centering\arraybackslash}X}

\input{ex_shared}

\ifpdf
\hypersetup{
  pdftitle={Functional Neural Wavefunction Optimization},
  pdfauthor={Victor Armegioiu, Juan Carrasquilla, Siddhartha Mishra, Johannes Müller, Jannes Nys, Marius Zeinhofer and Hang Zhang}
}
\fi

\myexternaldocument{ex_supplement}

\begin{document}

\maketitle

\begin{abstract}
We propose a framework for the design and analysis of optimization algorithms in variational quantum Monte Carlo, drawing on geometric insights into the corresponding function space. 
The framework translates infinite-dimensional optimization dynamics into tractable parameter-space algorithms through a Galerkin projection onto the tangent space of the variational ansatz. 
This perspective unifies existing methods such as stochastic reconfiguration and Rayleigh-Gauss-Newton, provides connections to classic function-space algorithms, and motivates the derivation of novel algorithms with geometrically principled hyperparameter choices.
We validate our framework with numerical experiments demonstrating its practical relevance through the accurate estimation of ground-state energies for several prototypical models in condensed matter physics
modeled with neural network wavefunctions. 
\end{abstract}



\section{Introduction}
The variational Monte Carlo (VMC) method is a flexible computational method for solving quantum many-body problems \emph{ab initio}. 
It tackles the stationary Schrödinger equation $\hat H\psi=E\psi$ directly via energy minimization, meaning the minimization of the Rayleigh quotient
\begin{equation}\label{eq:rayleigh_unconstrained}
    \min_{\psi\neq 0} E(\psi) 
    =
    \frac{\langle \psi,\hat H\psi\rangle}{\langle \psi, \psi\rangle}.
\end{equation}
Together with the Rayleigh-Ritz variational principle, which guarantees $E(\psi) \geq E_0$, this yields the necessary ingredients to approximate the ground state $\psi_0$ and the associated energy $E_0$
. 
VMC relies on stochastic integration via Monte Carlo sampling, allowing for highly expressive wavefunctions. Popular choices include hand-crafted Slater-Jastrow wavefunctions, as well as more generic neural network ansätze. Carleo and Troyer pioneered the latter approach in~\cite{carleo2017solving}, parameterizing the ground-state wavefunction of quantum spin systems with up to 100 spins.

Recent work has highlighted the capability of deep neural networks for approximating complex many‑body systems~\cite{wu2024variational, astrakhantsev2021broken, PhysRevResearch.2.023358,mossLeveragingRecurrenceNeural2025, pescia2024message, viteritti2023transformer, chen2024empowering, denis2025accurate, linteau2025phase}. 
For electronic systems, neural network architectures have been designed that conform to the Fermi-Dirac statistic. These neural networks guarantee particle-permutation symmetry of the fermionic wavefunction, with examples being FermiNet~\cite{pfau2020ab} and PauliNet~\cite{hermann2020deep}, and more recent architectures~\cite{pescia2024message, foster2025ab, von2022self, scherbela2024towards, li2024computational, robledo2022fermionic, kim2024neural, lou2024neural, pfau2024accurate}. These have achieved promising accuracy on small isolated molecules through VMC training.

Although VMC is flexible, exhibits favorable polynomial scaling with system size~\cite{hermann2023ab}, and, as an ab initio method, provides full access to the wavefunction, its advantages are offset by a challenging optimization landscape: high-dimensional, nonconvex, and stochastic. A popular optimization approach within the VMC community is stochastic reconfiguration~\cite{sorella1998green} which preconditions gradient updates by the inverse of the overlap matrix and significantly outperforms common deep learning optimizers such as Adam. The increased interest in deep neural network representations of quantum states has induced a revival of interest in alternative optimization algorithms~\cite{chen2024empowering, rende2024simple, neklyudov2023wasserstein}, ideally with faster convergence than stochastic reconfiguration~\cite{webber2022rayleigh, goldshlager2024kaczmarz, drissi2024second, peng2025analysis}.

In this work, we present a unified framework that integrates many of the diverse existing optimization algorithms for VMC. 
Following the recently proposed \emph{optimize-then-discretize} paradigm \cite{mullerposition, muller2023achieving}, we adopt an infinite-dimensional viewpoint to identify promising optimization algorithms for VMC. 
We take this approach beyond the instances known in the literature with stochastic reconfiguration arising from $L^2$ gradient descent and the recently proposed Wasserstein quantum Monte Carlo method~\cite{neklyudov2023wasserstein} arising from Wasserstein gradient descent. 
We base our functional setting on a Riemannian reformulation of the energy \eqref{eq:rayleigh_unconstrained} by restricting the admissible wavefunctions to the sphere $\mathbb{S} = \{ \psi  : \langle \psi, \psi \rangle = 1 \}$, resulting in a minimization problem over $\mathbb S$ given by 
\begin{equation}\label{eq:rayleigh_sphere}
    \min_{\psi\in\mathbb S}E(\psi)
    =
    \langle \psi, \hat H \psi \rangle.
\end{equation}
This perspective naturally motivates Riemannian optimization algorithms. 
We stress that this is more than mere nomenclature as a comparison of the non-Riemannian and Riemannian Newton methods reveals, see Remark~\ref{remark:newton_for_rq_ill_suited}.  
To minimize, we consider algorithms that start at $\psi_0\in \mathbb S$ and result in an iterative procedure of the form
\begin{equation}\label{eq:general-FS-update}
    \psi_{k+1} 
    =
    \frac{\psi_k - \eta_k d_{k}}{\lVert \psi_k - \eta_k d_{k} \rVert},
\end{equation}
adding directions $d_k$ scaled by a step-size $\eta_k$ and retracting back to the sphere. 
We stress that explicit normalization is typically not needed, ensuring flexibility in the choice of ansatz classes.
Note that some algorithms are naturally derived from a flow, in which cases the above scheme corresponds to an explicit time discretization. 
To turn the above into a computational scheme in parameter space, we project the directions $d_k$ onto the tangent space of the neural network ansatz $\psi_\theta$. 
We rely on \emph{Galerkin projections}, described in detail in \Cref{sec:Galerkin_discretization}. 
This leads to preconditioned gradient schemes of the form
\begin{equation}\label{eq:Galerkin_projected_algo}
    \theta_{k+1} = \theta_k - \eta_k Q(\theta_k)^{-1} \nabla_\theta E(\psi_{\theta_k}),
\end{equation}
where the choice of the matrix $Q(\theta_k)$ relates to the specific optimization algorithm chosen on the infinite-dimensional level. As concrete examples of functional algorithms, we discuss Riemannian gradient descent and Riemannian Newton methods. We show that Riemannian gradient descent, transferred to parameter space via Galerkin projections, leads to the well-known stochastic reconfiguration (SR) method. Our viewpoint allows a detailed convergence analysis of SR on the functional level, characterizing deteriorating convergence speed for closing spectral gaps, and a derivation of optimal step-sizes, for details we refer to \Cref{sec:SR}. 

As a second-order method, we consider a Riemannian Newton method, and show its correspondence to the classic Rayleigh quotient iteration on the functional level. We derive its Galerkin projected version of the form \eqref{eq:Galerkin_projected_algo} in \Cref{sec:Newton}. To globalize the local convergence of the Riemannian Newton method, we shift the Riemannian Hessian by a suitably scaled identity, possibly incorporating a priori information on the ground-state energy. We show that this corresponds to classic (shifted) inverse iteration on the functional level and we derive its Galerkin-projected version in \Cref{sec:inverse_iter}. We refer to it as \emph{projected inverse iteration (PII)} and identify it as a particularly well-suited choice for energy minimization. The decisive property of PII is its robustness to closing spectral gaps between the ground state and first excited state. This can be understood through the analysis of inverse iteration on the infinite-dimensional level which reveals that convergence deterioration in the presence of narrow spectral gaps can be mitigated by a suitable choice of shift parameter. 

\begin{table}
    \label{tab:algorithms}
    \centering
    \renewcommand{\arraystretch}{1.6}  
    \begin{tabularx}{\textwidth}{|C|C|C|}
    \hline
    {\raisebox{-0.8\height}{\textbf{Parameter algorithm}}} & \textbf{Function space optimizer} & {\raisebox{-0.8\height}{\textbf{Eigensolver}}} \\
    \hline\hline
    Projected inverse iteration & Globalized Riemannian Newton method on $\mathbb{S}$ & {\raisebox{-.8\height}{Shifted inverse iteration}} \\\hline
    Rayleigh Gauss-Newton \cite{webber2022rayleigh} & Riemannian Newton method on $\mathbb{S}$ & Rayleigh quotient iteration \\\hline
    Stochastic reconfiguration & Riemannian $L^2$ gradient descent on $\mathbb{S}$ & {\raisebox{-.8\height}{Power iteration}} \\\hline
    Wasserstein quantum Monte Carlo \cite{neklyudov2023wasserstein} & Wasserstein gradient descent & {\raisebox{-\height}{--}} \\\hline
    \end{tabularx}
    \caption{Correspondence of optimization algorithms across parameter space, function space, and---where applicable---eigensolvers. }
\end{table}

Numerically, we compare stochastic reconfiguration to the derived projected inverse iteration for several quantum spin systems on one- and two-dimensional lattices, see \Cref{sec:computational-examples}.
The numerical results confirm our theoretical findings, showing improved convergence properties of projected inverse iteration, especially in situations of closing spectral gaps. 
To conclude we summarize our main contributions:

\begin{itemize}
    \item \textbf{We propose a framework for translating infinite-dimensional optimization algorithms to tractable schemes in parameter space} by projecting functional updates onto the tangent space of the parametric ansatz. The Galerkin projections ensure that key properties of the functional algorithm are preserved in the discretized form. An overview is given in Table~\ref{tab:algorithms}.
    \item \textbf{We exercise the framework for Riemannian gradient and Riemannian Newton methods.} We show that  $L^2$ gradient descent on the sphere corresponds to stochastic reconfiguration within the framework, explaining its linear convergence and sensitivity to closing spectral gaps. We analyze (globalized) Riemannian Newton schemes, relate them to Rayleigh quotient and inverse iteration on the functional level, and derive their projected versions in parameter space.
    \item \textbf{We propose projected inverse iteration (PII),}
    a novel algorithm particularly well-suited for energy minimization. We derive it from a globalized Riemannian Newton method, corresponding to inverse iteration on the functional level. It can incorporate a priori information on the ground state energy, is robust to narrow spectral gaps, and can converge much faster than SR, as we demonstrate on examples on several quantum spin systems. 
\end{itemize}

\section{Notation and preliminaries}\label{sec:preliminaries}
We restrict our discussion to real wavefunctions. By $\mathcal H$ we denote the Hilbert space of quantum states, and for $v,w\in\mathcal H$ we denote their inner product by $\langle v, w\rangle$ and the corresponding norm by $\| v\|^2 = \langle v, v\rangle$.
The Hilbert space can be either finite or infinite-dimensional, with  $L^2(\mathbb R^{3N})$
and $\mathbb R^{2^N}$ being typical examples, where $N\in\mathbb N$ denotes the number of particles. 
The topological dual of a Hilbert space $\mathcal H$ is denoted by $\mathcal H^*$ and consists of the continuous linear functionals defined on $\mathcal H$ with values in $\mathbb R$. 
The Riesz representation theorem guarantees that every functional $f\in\mathcal H^*$ can be uniquely represented by a vector $v\in \mathcal H$ via $\langle \cdot,v\rangle = f$. 
The Hamiltonian $\hat H$ is a densely defined linear self-adjoint operator $\hat H\colon\operatorname{dom}(\hat H)\subset \mathcal H \to \mathcal H$. 
We denote its eigenvalues (in the point spectrum) in ascending order by $E_0, E_1, \ldots\in\mathbb R$ with corresponding eigenfunctions $\psi_0, \psi_1, \ldots \in\operatorname{dom}(\hat H)$. 
We assume that the spectrum of $\hat H$ is lower bounded; further assumptions on the Hamiltonians will be stated when relevant.

As we consider the problem \eqref{eq:rayleigh_sphere} of minimizing the Rayleigh quotient over the sphere, we work with the tangent space of $\mathbb{S}$ at a point $\psi\in \mathbb{S}$, which is given by
\begin{equation}
    T_\psi \mathbb{S} = \{ v \in \mathcal H : \langle \psi, v \rangle = 0 \}.
\end{equation}
The first two derivatives of the energy at $\psi\in \operatorname{dom}(\hat H)\setminus\{0\}$  are given by
\begin{equation*}
    \begin{split}
    DE(\psi)(v)
    &=
    \frac{2}{\|\psi\|^2}
    \left[
        \langle \hat H \psi, v \rangle
        - E(\psi)
        \langle  \psi, v  \rangle
    \right]
    \\
    D^2E(\psi)(v,w) 
    &= 
    \frac{2}{\|\psi\|^2}
    \left[
        \langle \hat H v, w \rangle
        - E(\psi) 
        \langle  v, w\rangle
        -
        \langle  \psi, w\rangle
        DE(\psi)(v) - 
        \langle  \psi, v\rangle
        DE(\psi)(w)
    \right]
    \end{split}
\end{equation*}
The corresponding differential and Hessian on the sphere act only on the tangent space $T_\psi\mathbb S$. 
For $\psi\in\mathbb S$ and tangent vectors $v, w\in T_\psi\mathbb S$ they are given by
\begin{align}
        DE(\psi)(v) 
        &= 2 
        \langle \hat H \psi, v \rangle
        , \label{eq:first-derivative-RQ} \\
        D^2E(\psi)(v, w) 
        &= 2 \left[ 
        \langle \hat H v, w \rangle
        - E(\psi)
        \langle  v, w\rangle
        \right], \label{eq:second-derivative-RQ}
\end{align}
see \cite[page 96, Example 5.17]{boumal2023introduction} for a derivation. From here on, $DE$ and $D^2E$ denote the Riemannian differential and Riemannian Hessian.

\subsection{Galerkin discretization and neural networks}\label{sec:Galerkin_discretization}
As stated above, to solve $\min_{\psi\in\mathbb S}E(\psi)$, we consider algorithms of the form
\begin{equation*}
    \psi_{k+1} = \frac{\psi_k - \eta_k d_{k}}{\lVert \psi_k - \eta_k d_{k} \rVert}.
\end{equation*}
For simplicity, we suppress the subscript $k$ below. In all functional algorithms that we consider, the direction $d\in T_{\psi} \mathbb S$ is given by the solution of a linear system posed in the tangent space 
\begin{equation}\label{eq:Galerkin_discretizations}
    b(d, v) = f(v), \quad\text{for all }v\in T_{\psi}\mathbb S.
\end{equation}
Here, $b$ is a bilinear form on $T_{\psi}\mathbb S$, $f$ is a member of the cotangent space $T_{\psi}^*\mathbb S$ and the solution $d$ is sought in $T_{\psi}\mathbb S$. Typical choices for $b$ are the $L^2$ inner product or the Riemannian Hessian $D^2E(\psi)$, and for $f$ the Riemannian differential $DE(\psi)$.

To approximate the solution $d$, a Galerkin discretization designates a finite-dimensional subspace $\mathcal V$, say spanned by the vectors $v_1,\dots v_P\in T_\psi\mathbb S$, where an approximate solution $d_\mathcal V$ to \eqref{eq:Galerkin_discretizations} is sought as a linear combination of $v_1,\dots v_P$ via
\begin{equation}\label{eq:galerkin-function-space}
    b(d_\mathcal V,v) = f(v)
\quad \text{for all } v\in \mathcal V.
\end{equation}
This leads to a linear system $B\mathrm \alpha = \mathrm f$ with $B_{ij}=b(v_j, v_i)$ and $\mathrm f_i = f(v_i)$. 
If it exists, the Galerkin solution is given by $d_\mathcal V = \sum_{i=1}^P \alpha_i v_i$. 
If $b$ is an inner product, the Riesz representation theorem guarantees the existence of both $d$ and $d_\mathcal V$ and implies that $d_\mathcal V$ is the $b$-orthogonal projection of $d$ onto span of $\{v_1,\dots, v_P\}$. In formulas this means
\begin{equation}\label{eq:lem:cea}
    \|d - d_\mathcal V\|_b = 
\inf_{v\in \mathcal V} \|d - v\|_b, 
\end{equation}
where $\|v\|_b = b(v,v)^\frac12$. 

The discussion so far was purely functional and did not require a parametric ansatz, and we now introduce a neural network ansatz $\psi_\theta$ with trainable parameters $\theta\in\Theta=\mathbb R^P$. 
We denote the normalized model by $\hat{\psi}_\theta = \psi_\theta / \lVert \psi_\theta \rVert$, which is often infeasible to compute and hence neural network-based ansätze are typically not \emph{explicitly} normalized. 
Rather, one makes use of the normalization implicitly by expressing the quantities of interest as expectations with respect to the Born density $p_\theta\propto \psi_\theta^2$ as we discuss in~\Cref{subsec:preliminaries-algorithmic}.

The objective or loss function used to optimize the network's parameters is the energy evaluated at a parametric wavefunction and given by 
\begin{equation}
    L(\theta) 
    =
    E(\hat \psi_\theta)
    =
    \frac{\langle  \psi_\theta, \hat H\psi_\theta \rangle}{\langle \psi_\theta, \psi_\theta\rangle}.
\end{equation}
We construct a Galerkin subspace of $T_{\psi_\theta} \mathbb S$ via 
\begin{equation}\label{eq:galerkin-space}
    \mathcal V_\theta
    = 
    \left\{ 
        \partial_{\theta_i}\hat{\psi}_\theta : i =1, \dots, P 
    \right\}
    \subset
    T_{\psi_\theta}\mathbb S,
\end{equation}
which corresponds to a generalized tangent space onto the model $\{\hat{\psi}_\theta:\theta\in\Theta\}$ parameterized by the network. 
We explicitly write out the Galerkin discretizations of the Riemannian differential $DE$ as well as the bilinear forms $\langle \cdot, \cdot \rangle$ and the Hamiltonian $\langle \cdot, \hat H\cdot \rangle$. 
When discretizing the differential, we obtain
\begin{equation}
    DE(\hat\psi_\theta)(\partial_{\theta_i} \hat \psi_\theta) 
    = 
    \partial_{\theta_i} L(\theta) 
    =
    2 \langle  \hat \psi_\theta, \hat H \partial_{\theta_i} \hat \psi_\theta \rangle 
\end{equation}
and hence $DE(\psi)$ corresponds to $\nabla L(\theta)$ upon Galerkin discretization in the space $\mathcal V_\theta$. Further, discretization of the $L^2$ inner product yields a matrix $S(\theta)$ with entries 
\begin{equation}\label{eq:overlap_matrix}
    S(\theta)_{ij} 
    =
    \langle \partial_{\theta_i} \hat\psi_\theta, \partial_{\theta_j} \hat\psi_\theta \rangle,
\end{equation}
which is commonly referred to as the \emph{overlap matrix} or \emph{quantum geometric tensor}~\cite{stokes2020quantum}. 
Finally, discretization of the Hamiltonian yields 
\begin{equation}
    H(\theta)_{ij} 
    =
    \langle  \partial_{\theta_i}\hat\psi_\theta, \hat H \partial_{\theta_j} \hat\psi_\theta \rangle, 
\end{equation}
see also \cite{toulouse2007optimization}. Note that the bilinear form induced by the Hamiltonian $\langle\cdot, \hat H\cdot \rangle$  is indefinite, if the Hamiltonian $\hat H$ is. 
Consequently, the matrix $H(\theta)$ is not necessarily definite.

The Galerkin discretization of all considered algorithms of the form \eqref{eq:general-FS-update}, leads us to preconditioned gradient schemes of the form
\begin{equation}
    \theta_{k+1} = \theta_k - \eta_k Q(\theta_k)^{-1} \nabla L(\theta_k),
\end{equation}
where the matrix $Q(\theta_k)\in\mathbb R^{P\times P}$ depends on the specific choice of algorithm.
Note that we can frequently reduce the cubic complexity $\mathcal O(P^3)$ of the matrix inversion to linear complexity in the dimension $P$ of the parameter space, relying on Woodbury-type matrix identities~\cite{rende2024simple, chen2024empowering}. Details are given in \Cref{subsec:woodbury}.

\subsection{Monte Carlo estimators}\label{subsec:preliminaries-algorithmic}
In practice, all quantities need to be estimated, as flexible wavefunction ansätze rule out exact computation of the inner product via quadrature rules. 
The idea is to formulate quantities as expectations with respect to the Born density $p_\theta \propto \psi_\theta^2$ and to apply Monte Carlo methods. Hence, let $(x_i)_{i=1,\dots,M}$ denote $M$ independent and identically distributed samples from $p_\theta$. In the following we use the abbreviation $\langle f \rangle = \mathbb E_{p_\theta}[f]$. It is convenient to introduce the \emph{local energy}, given by 
\begin{equation}
    E_L
    =
    E_L(\psi_\theta)
    =
    \frac{\hat H \psi_\theta}{\psi_\theta}.
\end{equation}

\paragraph{Estimating $L(\theta)$} The loss function can be written as the expectation of the local energy with respect to the Born density and approximated using the samples $(x_i)_{i=1,\dots,M}$, in formulas
\begin{equation}
    L(\theta)
    =
    \langle E_L(\psi_\theta) \rangle
    \approx
    \frac1M \sum_{i=1}^M E_L(\psi_\theta)(x_i).
\end{equation}
Given successful optimization, the integrand $E_L(\psi_\theta)$ tends to a constant function---the ground-state energy---and thus its variance vanishes. 

\paragraph{Estimating $\nabla L(\theta)$} Using that $\hat H$ is self-adjoint, we can write the gradient of the loss as
\begin{equation*}
    \nabla L(\theta)
    =
    \Big\langle
        2\big( 
            E_L(\psi_\theta) - L(\theta) 
        \big)
        \big(
            \nabla_\theta \log|\psi_\theta|
            -
            \langle\nabla_\theta \log|\psi_\theta|\rangle
        \big)
    \Big\rangle
    \approx
    O^Tr
\end{equation*}
where $r\in\mathbb R^M$ and $O\in\mathbb R^{M\times P}$ are given by
\begin{align}
    r_i \label{eq:residual}
    &=
    \frac{2}{\sqrt M}
    \left(
    E_L(\psi_\theta)(x_i) - \frac1M\sum_{k=1}^M E_L(\psi_\theta)(x_k)
    \right),
    \\
    O_{ij} \label{eq:O_matrix}
    &=
    \frac{1}{\sqrt M}
    \left[
        \partial_{\theta_j}\log|\psi_\theta|(x_i)
        -
        \frac1M \sum_{k=1}^M\partial_{\theta_j} \log|\psi_\theta|(x_k)
    \right].
\end{align}
Here $O\in\mathbb R^{M\times P}$ contains the centered log derivatives evaluated on the Monte Carlo samples. Typically, wavefunctions are parametrized in log-space in VMC. 
Note again that $E_L(\psi_\theta) - L(\theta) \to 0$ as $\psi_\theta$ approaches an eigenstate, resulting in a vanishing variance principle for $\nabla L(\theta)$.

\paragraph{Estimating $S(\theta)$} To estimate the overlap matrix $S(\theta)\in\mathbb R^{P\times P}$ defined in \eqref{eq:overlap_matrix}, we note first that the  generating functions of the Galerkin space $\mathcal V_\theta$ are given by 
\begin{equation*}
    \partial_{\theta_i} \hat{\psi}_\theta 
    = 
    \frac{1}{\| \psi_\theta \|}
    \left(
        \partial_{\theta_i}\psi_\theta
        -
        \frac{\langle\psi_\theta, \partial_{\theta_i}\psi_\theta\rangle}{\|\psi_\theta\|^2}\psi_\theta
    \right), \quad i=1,\dots,P.
\end{equation*}
A short computation shows that $S(\theta)$ can be expressed solely in terms of expectations with respect to the Born density as
\begin{equation*}\begin{split}
    S(\theta)_{ij}
    &=
    \Big\langle\big(
        \partial_{\theta_i} \log |\psi_\theta| - \langle\partial_{\theta_i} \log |\psi_\theta|\rangle
        \big)
        \big(
        \partial_{\theta_j} \log |\psi_\theta| - \langle\partial_{\theta_j} \log |\psi_\theta|\rangle
    \big)\Big\rangle
    \approx
    (O^\top O)_{ij}. 
\end{split}\end{equation*}

\paragraph{Estimating $H(\theta)$} 
Similar to the case of the overlap matrix, we can provide an estimator for the discretization of the Hamiltonian matrix~\cite{toulouse2007optimization,umrigar2007alleviation, webber2022rayleigh}. 
\begin{lemma}[Hamiltonian Estimator]
The estimator of the Hamiltonian matrix $H(\theta)$ is derived from the expression 
\begin{equation*}
\begin{split}
    H(\theta)_{ij}
    &=
    \Big\langle
        \left[
            E_{L,j} 
            +
            E_L
            \left(
                \partial_{\theta_j}\log|\psi_\theta| 
                -
                \left\langle \partial_{\theta_j}\log|\psi_\theta|\right\rangle
            \right)
        \right]
        \cdot
        \left[
            \partial_{\theta_i}\log|\psi_\theta| 
            -
            \left\langle \partial_{\theta_i}\log|\psi_\theta|\right\rangle
        \right]
    \Big\rangle,
\end{split}
\end{equation*}
where $E_{L,j} = \partial_{\theta_j}E_L$ denotes the derivative of the local energy. Its estimator on the samples $(x_i)_{i=1,\dots,M}$ is given by
\begin{equation}
    H(\theta) \approx
    O^\top A, \label{eq:H_matrix_estimator}
\end{equation}
where the matrix $A \in \mathbb{R}^{M \times P}$ is defined as 
\begin{equation*}
    A_{ij} 
    =
    \frac{1}{\sqrt{M}}\left[\partial_{\theta_j}E_L(\psi_\theta)(x_i) 
    + 
    E_L(\psi_\theta)(x  _i)O_{ij}\right]. 
\end{equation*}
Despite the symmetry of $H(\theta)$, 
the estimator is only 
asymptotically symmetric.
\end{lemma}
\begin{proof}
    Note that the derivative of the local energy is given by
    \begin{equation*}
        E_{L,j} = \frac{\hat H\partial_{\theta_j}\psi_\theta}{\psi_\theta} - \partial_{\theta_j}\log|\psi_\theta|E_L.
    \end{equation*}
    Hence, we compute
    \begin{align*}
        \hat H\partial_{\theta_j}\hat \psi_\theta
        &=
        \frac{1}{\|\psi_\theta\|}
        \left[
            \hat H\partial_{\theta_j}\psi_\theta 
            -
            \frac{\langle \psi_\theta, \partial_{\theta_j}\psi_\theta \rangle}{\|\psi_\theta\|^2}\hat H\psi_\theta
        \right]
        \\
        &=
        \frac{\psi_\theta}{\|\psi_\theta\|}
        \left[
            \frac{\hat H \partial_{\theta_j}\psi_\theta}{\psi_\theta}
            -
            \langle \partial_{\theta_j}\log|\psi_\theta|  \rangle E_L
        \right]
        \\
        &=
        \frac{\psi_\theta}{\|\psi_\theta\|}
        \left[
            E_{L,j}
            +
            \left(
                \partial_{\theta_j}\log|\psi_\theta|
                -
                \langle \partial_{\theta_j}\log|\psi_\theta|  \rangle
            \right) 
            E_L
        \right],
    \end{align*}
    which readily yields the asserted form of the estimator.
\end{proof}

\subsection{Solving the linear systems}\label{subsec:woodbury}
It is common practice to add a regularization with parameter $\varepsilon>0$ to the preconditioners. For stochastic reconfiguration, the regularized linear system that needs to be solved then takes the form 
\begin{equation*}
    (O^TO + \varepsilon I)\xi
    =
    O^Tr.
\end{equation*}
The system size is $P$, which can be in the millions, and the matrix $O^TO + \varepsilon I$ is dense making a direct solution method intractable. We work under the assumption $M\ll P$ and use the Woodbury identity to reformulate
\begin{equation*}
    (O^TO + \varepsilon I)^{-1}O^Tr
    =
    O^T(OO^T + \varepsilon I)^{-1}r.
\end{equation*}
Now, at every iteration, the matrix $OO^T$ needs to be computed, and a linear system in sample space needs to be solved. These operations have complexity $\mathcal O(M^2P)$ and $\mathcal O(M^3)$ respectively, instead of the cubic cost $\mathcal O(P^3)$ in the original formulation. 
This equivalent reformulation is commonly referred to as \emph{minSR} in the context of stochastic reconfiguration~\cite{chen2024empowering, rende2024simple}. The structure $O^\top A$ of the estimator of $H(\theta)$ allows for similar formulations, which we use in the simulations in \Cref{sec:computational-examples}. For the Riemannian Newton methods, we make use of the following formula
\begin{equation*}
    (O^\top A - \tau O^\top O + \varepsilon I)^{-1} O^\top r
    =
    O^\top(A O^\top - \tau O O^\top + \varepsilon I)^{-1} r,
\end{equation*}
which holds for any $\tau\in\mathbb R$.

\section{A Unified Framework for Energy Minimization}\label{sec:functional-framework}
This section discusses functional first- and second-order methods, more precisely, a Riemannian $L^2$ gradient descent on the sphere and a Riemannian Newton method for the minimization of the Rayleigh quotient. 
Identifying stochastic reconfiguration (SR) as an $L^2$ gradient descent allows us to analyze its convergence speed in the small step size limit on the functional level and derive optimal step sizes for the explicit Euler discretization. 
Subsequently, we discuss Newton's method on the sphere and show that it corresponds to Rayleigh quotient iteration in function space, which is known to have favorable convergence properties\footnote{Rayleigh quotient iteration converges locally at a cubic rate, at least in the finite dimensional case $\mathcal H = \mathbb C^{2^N}$. For unbounded operators, we are not aware of a proof in the literature.}
\cite{trefethen2022numerical}. 
Finally, we discuss a globalization approach of the Riemannian Newton method by a diagonal shift in function space and show that it corresponds to shifted inverse iteration on the level of wave functions. 
This allows us to leverage a priori information on the ground-state energy and yields a highly effective method for energy minimization---which we call \emph{projected inverse iteration (PII)}. 
We compare the empirical performance of PII and SR in \Cref{sec:computational-examples}.

\subsection{\texorpdfstring{$L^2$} G Gradient Descent}\label{sec:SR}
This section discusses energy minimization via $L^2$ gradient descent on the sphere which leads to the well-known stochastic reconfiguration scheme. We begin with a concise derivation and discuss practical implications. 

\subsubsection*{Synopsis}
To perform energy minimization, consider the following functional scheme that starts with a given $\psi_0\in\mathbb S$ and proceeds via the update rule
\begin{equation}\label{eq:L2-GD-function-space}
    \psi_{k+1} 
    =
    \frac{\psi_k - \eta_k \nabla E(\psi_k)}{\|\psi_k - \eta_k \nabla E(\psi_k)\|}.
\end{equation}
The update direction is given by the gradient $\nabla E(\psi_k) = 2[\hat H\psi_k - E(\psi_k)\psi_k]$ with respect to the $L^2$ metric restricted to $\mathbb S$, see~\eqref{eq:first-derivative-RQ}. 
By definition, $\nabla E(\psi_k)$ is characterized as the unique solution $d_k\in T_{\psi_k}\mathbb S$ to the linear equation
\begin{equation}\label{eq:linear_equation_for_sr_direction}
    \langle d_k, v \rangle 
    =
    DE(\psi_k)[v]
    =
    2 \langle \hat H \psi_k, v \rangle
    \quad \text{for all } v\in T_{\psi_k} \mathbb S. 
\end{equation}
In light of \Cref{sec:Galerkin_discretization}, this equation can be used to perform a Galerkin discretization of the functional algorithm \eqref{eq:L2-GD-function-space} in the space $\mathcal V_\theta$ spanned by the derivatives of the normalized wavefunctions, see \eqref{eq:galerkin-space}, 
where $\psi_k$ is parametrized as $\psi_k=\psi_{\theta_k}$. 
This leads to the overlap matrix  $S(\theta_k)_{ij} = \langle \partial_{\theta_i}\hat\psi_{\theta_k}, \partial_{\theta_j}\hat\psi_{\theta_k} \rangle$ for the left-hand side of equation \eqref{eq:linear_equation_for_sr_direction} and the energy gradient $2\langle \hat H \hat \psi_{\theta_k}, \partial_{\theta_i}\hat\psi_{\theta_k} \rangle = \partial_{\theta_i} L(\theta_k)$ stemming from the right-hand side. Introducing a regularization parameter $\varepsilon_k \geq 0$, we obtain
\begin{equation}\label{eq:SR}
    \theta_{k+1} = \theta_{k} - \eta_k (S(\theta_k) + \varepsilon_k I)^{-1}\nabla L(\theta_k). 
\end{equation}
The preconditioned gradient scheme~\eqref{eq:SR} is known as \emph{stochastic reconfiguration (SR)} \cite{sorella2001generalized} or \emph{quantum natural gradient descent} and is widely used in practice~\cite{carleo2017solving, lange2024architectures, medvidovic2024neural}. 

Without regularization, we obtain an update direction $\xi$ satisfying $S(\theta)\xi = \nabla L(\theta)$ and therefore the function space update can be approximated according to 
\begin{equation}
    \hat\psi_{\theta-\eta \xi} 
    =
    \hat\psi_{\theta} 
    -
    \eta \sum_{i=1}^P \xi_i \partial_{\theta_i} \hat\psi_\theta + O(\eta^2) = 
    \hat\psi_{\theta} - \eta d_{\mathcal V_\theta} + O(\eta^2). 
\end{equation}
By Galerkin orthogonality, we immediately obtain that $d_{\mathcal V_\theta}$ is the $L^2$ orthogonal projection of $\nabla E(\psi_\theta)$ onto the Galerkin space $\mathcal V_\theta\subset T_{\psi_\theta}\mathbb S$, see~\Cref{eq:lem:cea}.
In words, stochastic reconfiguration with small step sizes leads to function space updates approximating the projection of the $L^2$ gradient onto the Galerkin space. 
Assuming expressive spaces $\mathcal V_\theta$, we can hence study the non-parametric scheme \eqref{eq:L2-GD-function-space} to gain insight into the optimization dynamics of SR. 
We will do this in the remainder of this subsection, where, in particular, we show the following: 
\begin{itemize}
    \item 
    In the small step size limit, the dynamics of \eqref{eq:L2-GD-function-space} converge to a ground state at a linear rate, that is, at an exponential speed, with exponent given by the spectral gap of $\hat H$, see \Cref{thm:convergence_speed}. 
    \item We explicitly compute the optimal step sizes for \eqref{eq:L2-GD-function-space} in function space. 
\end{itemize}

The small step size limit of~\eqref{eq:L2-GD-function-space} agrees up to a constant scaling of time with the so-called imaginary-time Schrödinger flow, 
which yields an explicit relation of the convergence to the Hamiltonian spectrum as shown in the next section~\cite{sorella1998green, sorella2001generalized, becca2017quantum, yuan2019theory}.

\subsubsection*{Mathematical Derivations}
To understand the convergence behavior of the gradient descent on the sphere given in \eqref{eq:L2-GD-function-space}, we consider its small step size limit. 
Hence, for an arbitrary initial value $\psi_0\in\mathbb S$, we consider a solution $\psi\colon[0,\infty)\to\mathbb S \cap \operatorname{dom}(\hat H)$ of the $L^2$ gradient flow on $\mathbb S$, which is described by   
\begin{equation}\label{eq:gradient_flow_sphere}
    \begin{split}
        \psi'(t) 
        &=
        - \nabla E(\psi(t)) 
        =
        -2 [ 
        \hat H\psi(t)  - E(\psi(t))\psi(t)],
        \\
        \psi(0) &= \psi_0.
    \end{split}
\end{equation}
This is equivalent to $\langle\psi'(t), v\rangle = -2\langle\hat H\psi(t), v\rangle$ for all $v\in T_{\psi(t)}\mathbb S$, where the time derivative is to be understood in the classical sense with values in $\mathcal H$. 

The solution to \eqref{eq:gradient_flow_sphere} can be explicitly described by normalization of the corresponding evolution on $\mathcal H$, meaning that $\psi(t)\propto e^{-2t\hat H} \psi_0$
if $\hat H$ generates a semigroup on $\mathcal H$. 
The un-normalized dynamics $e^{- 2t\hat H} \psi_0$ are---up to the constant scaling factor of $2$ in time---known as \emph{imaginary time evolution}~\cite{neuscamman2012optimizing}, to be compared to real-time quantum dynamics $e^{- i t\hat H} \psi_0$ describing the evolution of quantum states via the Schrödinger equation. 
We will see that although the convergence of $\psi(t)$ to a ground state is linear, the rate depends on the spectral gap $E_1-E_0$ and can thus deteriorate. 

\begin{theorem}[Continuous Dynamics]\label{thm:convergence_speed}
    We assume that $\hat H$ is a densely defined, self-adjoint linear operator on $\mathcal H$. 
    Further, we assume that $E_0 \coloneqq \inf \sigma(\hat H)>-\infty$ is an eigenvalue, and that $E_1 \coloneqq\inf (\sigma(\hat H)\setminus\{E_0\})>E_0$ and consider a solution $(\psi(t))_{t\ge0}$ of the gradient flow~\eqref{eq:gradient_flow_sphere}. 
    We denote the projection onto the eigenspace $V$ of $E_0$ by $P_0$, assume that the initial value 
    has overlap with the ground states, meaning $\psi_0\notin V^\perp$, and set $ v_0\coloneqq \frac{P_0 \psi_0}{\lVert P_0 \psi_0 \rVert}$. 
    Then it holds that 
    \begin{equation}
        \left\| v_0 - \psi(t)\right\|
        \leq
        \sqrt{2} \alpha e^{-2(E_1-E_0)t}, 
    \end{equation}
    where $\alpha=\frac{\lVert \psi_0-P_0\psi_0 \rVert}{\lVert P_0 \psi_0 \rVert}$. 
    Note that $\alpha>0$, unless $\psi_0\in V$ is already a ground state. 
\end{theorem}
\begin{proof}
First, note that $-2\hat H$ generates a semigroup $(e^{-2t \hat H})_{t\ge0}$ of self-adjoint operators as it is self-adjoint and has a spectrum that is bounded from above~\cite[Example II.3.27]{engel2000one}. 
For $\psi_0\in \mathcal H$ with unit norm, we consider the 
auxiliary function $\xi(t)= e^{-2t\hat H} \psi_0$. 
Setting
we compute that 
\begin{equation*}
\begin{split}
    \psi'(t) 
    &=
    \frac{\xi'(t)\|\xi(t)\| - \left\langle \frac{\xi(t)}{\|\xi(t)\|},\xi'(t) \right\rangle \xi(t)}{\|\xi(t)\|^2}
    \\ &
    =
    -\frac{2\hat H \xi(t)}{\|\xi(t)\|}
    +
    \left\langle \frac{\xi(t)}{\|\xi(t)\|}, \frac{2\hat H\xi(t)}{\|\xi(t)\|} \right\rangle\frac{\xi(t)}{\|\xi(t)\|}
    =
    -2[\hat H\psi(t) - E(\psi(t)) \psi(t)],
\end{split}
\end{equation*}
which shows that $\psi$ solves \eqref{eq:gradient_flow_sphere}. 

We consider the orthogonal decomposition $\mathcal H = V+W$, where $V$ is the eigenspace corresponding to $E_0$ and $W= V^\perp$, where we used that $V$ is closed as it is the eigenspace of a self-adjoint operator. 
Note that both $V$ and $W\cap\operatorname{dom} \hat H$ are invariant under $\hat H$, hence we can study the induced semigroup on the two subspaces as $e^{-2t\hat H} \psi_0 = e^{-2t\hat H} v_0 + e^{-2t\hat H} w_0$ for $ \psi_0 =  v_0 + w_0$, where $ v_0\in V$ and $ w_0\in W$. 
For $v\in V$ we have $e^{-2t\hat H} v=e^{-2tE_1}v$, hence, it remains to estimate $e^{-2t\hat H} w$ for $w\in W$. 
Thus, we want to understand the growth constant of the semigroup $e^{-2t\hat H}$ on $W$, which agrees with the spectral radius under the spectral mapping property. 
In our case, we first note that $\hat H|_W$ is self-adjoint on $W$ and thus the semigroup is self-adjoint and the spectral mapping theorem guarantees $\sigma(e^{-2t\hat H|_W}) = e^{-2t \sigma(\hat H)}$, see~\cite[Corollary 3.14]{engel2000one}. 
Therefore, we can use the spectral mapping theorem to see that 
\begin{equation*}
    \lVert e^{-2t\hat H|_W} \rVert 
    =
    \sup \sigma(e^{-2t\hat H|_W}) 
    =
    \sup e^{-2t \sigma(\hat H|_W)} 
    =
    e^{-2t \inf \sigma(\hat H|_W)} 
    =
    e^{-2E_1 t},
\end{equation*}
which guarantees $\lVert e^{-2t\hat H}  w_0 \rVert \le e^{-2E_1 t} \lVert  w_0 \rVert$. 

Using 
the Pythagorean theorem, we estimate  
\begin{align}\label{eq:help-convergence-l2-flow}
    \begin{split}
        \left\lVert \frac{v_0}{\lVert v_0 \rVert} - 
    \frac{e^{-2t\hat H} \psi_0}{\lVert e^{-2t\hat H} \psi_0\rVert}
    \right\rVert^2 
    & = 
    \left\lVert \frac{v_0}{\lVert v_0 \rVert} - \frac{e^{-2E_0 t} v_0}{\lVert e^{-2t\hat H} \psi_0\rVert} \right\rVert^2 
    +
    \left\lVert \frac{ e^{-2t\hat H} w_0}{\lVert e^{-2t\hat H} \psi_0\rVert^2} \right\rVert^2 \\ 
    & = 
    \left(1 - \frac{e^{-2E_0 t} \lVert v_0 \rVert}{\lVert e^{-2t\hat H} \psi_0\rVert}\right)^2
     + 
     \frac{\rVert e^{-2t\hat H} w_0\rVert^2 }{\lVert e^{-E_0 t} v_0 \rVert^2 + \lVert e^{-t \hat H} w_0 \rVert^2}.
    \end{split} 
\end{align}
We upper bound~\eqref{eq:help-convergence-l2-flow} via 
\begin{equation*}
    \begin{split}
    \left\lVert \frac{v_0}{\lVert v_0 \rVert} - \frac{e^{-2t\hat H} \psi_0}{\lVert e^{-2t\hat H} \psi_0\rVert}
    \right\rVert^2 
    & \le 
    \left(1 - \frac{e^{-2E_0 t} \lVert v_0 \rVert}{e^{-2E_0 t} \lVert v_0\rVert + e^{-2E_1 t} \lVert w_0 \rVert}\right)^2 
     +
     \left(\frac{ e^{-2E_1 t}\rVert w_0\rVert }{e^{-2E_0 t} \lVert v_0\rVert}\right)^2 \\ 
    & \le
    \frac{2 \lVert w_0 \rVert^2}{\lVert v_0 \rVert^2} \cdot e^{-4(E_1-E_0) t},
\end{split}
\end{equation*}
which concludes the proof. 
\end{proof}

\begin{remark}[Lower bound]
Note that we can also use~\eqref{eq:help-convergence-l2-flow} to establish a lower bound on the convergence speed of the $L^2$ gradient flow. 
Indeed, we can simply drop the first term and obtain 
\begin{equation*}
    \left\lVert v_1 - \psi(t) 
    \right\rVert^2 
    \ge 
    \frac{\rVert e^{-2t\hat H} w_0\rVert^2 }{\lVert  e^{-E_0 t} v_0 \rVert^2 + \lVert e^{-t \hat H} w_0 \rVert^2} \ge 
    \frac{\rVert e^{-2t\hat H} w_0\rVert^2 }{e^{-4E_0 t} \lVert \psi_0 \rVert^2} 
    =
    \frac{\rVert e^{-2t\hat H} w_0\rVert^2 }{e^{-4E_0 t}}.
\end{equation*}
Hence, the lower bound depends on the growth or decay of $\rVert e^{-2t\hat H} w_0\rVert$. 
Assume, for example, that $E_1$ is in the point spectrum and that the projection $P_1\psi_0$ onto the eigenspace $V_1$ of $E_1$ is non-trivial. 
Then, we obtain a lower bound of 
\begin{equation*}
    \left\lVert v_1 - \psi(t) 
    \right\rVert 
    \ge
    \frac{\lVert e^{-2t\hat H} \psi_0 \rVert}{e^{-2E_0 t}} 
    \ge
    e^{-2(E_1-E_0) t} \lVert P_0\psi_0 \rVert. 
\end{equation*}
\end{remark}

\Cref{thm:convergence_speed} indicates that stochastic reconfiguration converges at a linear rate with deteriorating convergence speed as the spectral gap 
closes. 
In practice, we use the explicit Euler discretization \eqref{eq:L2-GD-function-space}, for which the optimal step size can be computed explicitly by minimizing the energy on the affine span of $\psi$ and $d$. 

\begin{lemma}[Optimal Step Sizes]\label{lem:optimal_step_size}
    Let $\psi$ be a member of $\operatorname{dom}(\hat H)$ not necessarily normalized and let $d\in T_{\psi}\mathbb S$ denote the $L^2$ gradient of $E$ on the sphere.
    Then the optimal step size $\eta^*$
    \begin{equation*}
        \eta^* = \operatorname{argmin}_{\substack{\eta > 0}}E(\psi + \eta d)
    \end{equation*}
    is given by \begin{equation}\label{eq:optimal_step_sr}
        \eta^*
        =
        \frac{
        \|\psi\|}{E(d)-E(\psi) + \sqrt{(E(d) - E(\psi))^2+\|d\|^2}}.
    \end{equation}
\end{lemma}
\begin{remark}
    The formulation of the Lemma above is precisely the situation encountered in practice when the ansatz $\psi_\theta$ is not normalized, but thanks to Monte Carlo sampling the energy gradient and the overlap matrix are computed for the normalized ansatz. In this case \eqref{eq:optimal_step_sr} is the optimal step size for the functional algorithm. Note that $d$ is given by
    \begin{equation*}
        d
        =
        \nabla E(\psi/\|\psi\|)
        =
        \frac{2}{\|\psi\|}
        \left[
            \hat H\psi - E(\psi)\psi
        \right],
    \end{equation*}
    hence invariant under rescaling $\psi$. The same holds true for the remaining terms in the denominator, hence the scaling of the optimal step size $\eta^*$ is entirely determined by $\|\psi\|$. This scaling can serve as an explanation for the manual learning rate tuning needed for optimal SR performance~\cite{goldshlager2024kaczmarz}. 
\end{remark}
\begin{proof}
    Writing out the minimization problem via expanding the energy $E(\psi + \eta d)$ and using $\langle\psi, \psi\rangle=1$, $\langle \psi, d\rangle = 0$ and 
    $\langle \hat H\psi, d\rangle = -\frac12\|d\|^2$ we get that 
    \begin{equation*}
    \begin{split}
        \min_{\eta > 0} E(\psi + \eta d) 
        &=
        \min_{\eta > 0}
        \frac{
        \begin{pmatrix}
            1 & \eta
        \end{pmatrix}
        \begin{pmatrix}
            \langle\hat H\psi, \psi\rangle & -\frac12\|d\|^2\|\psi\| \\
            -\frac12\|d\|^2\|\psi\| & \langle\hat Hd, d\rangle
        \end{pmatrix}
        \begin{pmatrix}
            1 \\
            \eta
        \end{pmatrix}
        }
        {
        \begin{pmatrix}
            1 & \eta
        \end{pmatrix}
        \begin{pmatrix}
            \langle \psi, \psi \rangle & 0 \\
            0 & \langle d, d\rangle
        \end{pmatrix}
        \begin{pmatrix}
            1 \\
            \eta
        \end{pmatrix}
        }
        \\
        &=
        \min_{\eta > 0}
        \frac{
        \begin{pmatrix}
            1 & \eta
        \end{pmatrix}
        A
        \begin{pmatrix}
            1 \\
            \eta
        \end{pmatrix}
        }
        {
        \begin{pmatrix}
            1 & \eta
        \end{pmatrix}
        B
        \begin{pmatrix}
            1 \\
            \eta
        \end{pmatrix}
        }.
    \end{split}
    \end{equation*}
    We proceed by finding the eigenvectors and eigenvalues of $M = B^{-1}A$. The eigenvector of $M$ corresponding to the minimal eigenvalue 
    will then be normalized to contain unity in the first component in order to determine $\eta^*$. The matrix $M$ is given by
    \begin{equation*}
        M 
        =
        \begin{pmatrix}
            E(\psi)& -\frac12\frac{\|d\|^2}{\|\psi\|} \\
            -\frac12\|\psi\| & E(d)
        \end{pmatrix}.
    \end{equation*}
    We use that the eigenvalues of $M$ are given by
    \begin{equation*}
    \begin{split}
        \lambda_{1,2}
        &=
        \frac{\operatorname{tr}(M)}{2} \pm \sqrt{\frac14 \operatorname{tr}(M)^2 - \det(M)}
        \\
        &=
        \frac{E(\psi) + E(d)}{2} \pm \frac12\sqrt{ (E(\psi) - E(d))^2 + \|d\|^2 }.
    \end{split}
    \end{equation*}
    The minimal eigenvalue is thus $\lambda_2$ and a corresponding eigenvector is given by
    \begin{equation*}
        \begin{pmatrix}
            \lambda_2 - E(d) \\ -\frac12 \|\psi\|
        \end{pmatrix}.
    \end{equation*}
    Dividing by the first component, we get $\eta^* = -\frac12\|\psi\|/(\lambda_2 - E(d))$. 
\end{proof}

\paragraph{SR as quantum natural gradient}
An alternative interpretation of stochastic reconfiguration can be given in terms of (quantum) natural gradients~\cite{stokes2020quantum}. 
Here, one uses the Fisher-information matrix $F(\theta)$ of the statistical model $p_\theta \propto 
\psi_\theta^2
$ as a preconditioner, which is also known as the Fubini-Study metric. 
It turns out that the Fisher-information matrix $F(\theta)$ agrees with the overlap matrix $S(\theta)$. 
This can either be verified by direct computation or by using the fact that $\psi\mapsto \psi^2$ is a Riemannian isometry between the positive orthants of the $L^2$ and $L^1$ spheres with respect to the $L^2$ and Fisher-Rao metrics, respectively, see~\cite[pp. 128]{ay2017information}. 
We refer to~\cite{lange2024architectures} for a more detailed discussion of quantum natural gradients. 
While this interpretation works with the Born densities and connects stochastic reconfiguration to information geometry, it seems more natural for us to study the problem on the sphere $\mathbb S$.

\subsection{Riemannian Newton}\label{sec:Newton}
We have seen that stochastic reconfiguration arises from an $L^2$ gradient descent, where we now discuss a Riemannian-Newton method for the energy $E$ on the sphere $\mathbb S$. 
We prove that it corresponds to Rayleigh quotient iteration. 
Again, we provide a concise discussion before turning towards rigorous results. 

\subsubsection*{Synopsis}
Given an initial iterate $\psi_0\in\mathbb S \cap \operatorname{dom}(\hat H)$ 
Newton's method on the sphere with the retraction by normalization is given by
\begin{align}
    \label{eq:newton_sphere_1}
    d_k 
    &=
    D^2E(\psi_k)^{-1}[DE(\psi_k)],
    \\
    \label{eq:newton_sphere_2}
    \psi_{k+1}
    &=
    \frac{ \psi_k - d_k}{\| \psi_k - d_k\|}.
\end{align}
Here, the Riemannian Hessian $D^2E(\psi_k)$ and differential $DE(\psi_k)$ are only defined for tangent vectors, see \eqref{eq:first-derivative-RQ} and \eqref{eq:second-derivative-RQ}. Writing out equation \eqref{eq:newton_sphere_1} shows that the update direction $d_k$ solves the linear equation
\begin{equation}\label{eq:characterization-Newton-direction}
    \langle
        (\hat H - E(\psi_k)I)d_k, v
    \rangle
    =
    \langle 
        \hat H \psi_k, v
    \rangle
    \quad 
    \text{for all }v\in T_{\psi_k}\mathbb S.
\end{equation}
As described in \Cref{sec:Galerkin_discretization}, this equation can be used to Galerkin discretize the functional algorithm \eqref{eq:newton_sphere_1}\&\eqref{eq:newton_sphere_2} in the space $\mathcal V_\theta$ spanned by the derivatives of the normalized wavefunctions, see \eqref{eq:galerkin-space}, where $\psi_k = \hat \psi_{\theta_k}$. 
The Hamiltonian term leads to the Hamiltonian matrix $H(\theta_k)_{ij} = \langle \partial_{\theta_i}\hat\psi_{\theta_k}, \hat H\partial_{\theta_j}\hat\psi_{\theta_k} \rangle$ and the identity yields the overlap matrix $S(\theta_k)$. 
The discretization of the right-hand side results in $\frac12 \nabla L(\theta_k)$. 
As~\eqref{eq:newton_sphere_1} is a Newton scheme, it does not incorporate a step size. 
When using a nonlinear ansatz $\psi_\theta$, we employ a linearization in the parameters for the discretization. 
As this is only valid locally, we introduce a step size to account for the nonlinearity of the ansatz. 
Adding regularization $\varepsilon_k$, we obtain the preconditioned gradient scheme 
\begin{equation}\label{eq:discretized_newton_scheme}
    \theta_{k+1}
    =
    \theta_k
    -\frac{\eta_k}{2}
    \left(
    H(\theta_k) - L(\theta_k)S(\theta_k) + \varepsilon_k I
    \right)^{-1}\nabla L(\theta_k).
\end{equation}

Note that the bilinear form $D^2(\psi)(u,v) = \langle (\hat H - E(\psi) I) u, v \rangle$ appearing in \eqref{eq:characterization-Newton-direction} can be indefinite. 
Hence, Galerkin orthogonality and C\'ea's lemma do not automatically hold. 
However, we show in \Cref{lem:cea-newtons-method} in the following subsection that if the closest eigenvalue to $E(\psi_{\theta_k})$ is the ground state energy $E_0$, then 
$D^2E(\psi_{\theta_k})$
is indeed an inner product on $T_{\psi_{\theta_k}}\mathbb S \cap \operatorname{dom}(\hat H)$ and consequently the matrix $H(\theta_k) - L(\theta_k)S(\theta_k)$ is positive semi-definite. 
Moreover, this guarantees that 
for any $\xi$ solving $(H(\theta_k) - L(\theta_k)S(\theta_k))\xi = -\frac12\nabla L(\theta_k)$ the function space update $\sum_{i=1}^P\xi_i\partial_{\theta_i}\hat \psi_\theta$ is the orthogonal projection of the Newton direction $d_k$ to the space $\mathcal V_\theta$ with respect to the inner product induced by $\hat H - L(\theta_k)I$. 
Hence, close to the ground state we can think of \eqref{eq:discretized_newton_scheme} as a projected Riemannian Newton method. 

As a complementary perspective on \eqref{eq:newton_sphere_1}\&\eqref{eq:newton_sphere_2}, we offer an interpretation as a textbook eigenvalue algorithm. 
As Lemma~\ref{lem:newton_direction_representation} shows, it holds
\begin{equation}
    \frac{\psi_k - d_k}{\| \psi_k - d_k \|}
    =
    \frac{(\hat H - E(\psi_k)I)^{-1}\psi_k}{\| (\hat H - E(\psi_k)I)^{-1}\psi_k \|},
\end{equation}
which is known as \emph{Rayleigh quotient iteration}, converging with a 
locally cubic rate~\cite{trefethen2022numerical}. 
We make the correspondence between the Riemannian Newton method and Rayleigh quotient iteration precise in \Cref{lem:newton_direction_representation}.

\subsubsection*{Mathematical Derivations}
We provide a best-approximation result for the Galerkin discretizations of the Riemannian Newton method, and an interpretation as Rayleigh quotient iteration.

\begin{lemma}\label{lem:cea-newtons-method}
    We assume that $\hat H$ is a densely defined, self-adjoint linear operator on $\mathcal H$, that $E_0 \coloneqq \inf \sigma(\hat H)>-\infty$ is an eigenvalue, and that it holds $E_1 \coloneqq\inf (\sigma(\hat H)\setminus\{E_0\})>E_0$. Additionally, we assume that the ground state has multiplicity one and we let $\psi\in\mathbb S\cap\operatorname{dom}(\hat H)$ be such that $E(\psi)<\frac{E_0+E_1}{2}$. Then $a(v,w)= \langle \hat Hv, w\rangle - E(\psi)\langle v, w\rangle$ is an inner product on $T_\psi\mathbb S\cap \operatorname{dom}(\hat H)$. Consequently, for $d=D^2E(\psi)^{-1}[DE(\psi)]$ and its Galerkin discretization $d_\mathcal V$ the following best approximation holds:
    \begin{equation}\label{eq:quasi_best_approx}
        \|d_{\mathcal V} - d\|
        \le \alpha^{-1}
        \| d_{\mathcal V} - d \|_a 
        =
        \inf_{q\in\mathcal V_\theta}\| q - d \|_a, 
    \end{equation}
    where the $a$-norm is given by $\|\cdot\|_a=\sqrt{a(\cdot, \cdot)}$ and 
    $E(\psi)=\frac{E_0+E_1}{2}-\alpha$. 
\end{lemma}
\begin{proof}
Note that it suffices to show $\langle v, \hat Hv\rangle \ge \frac{E_0+E_1}{2}>E(\psi)$ for any tangent vector $v\in T_\psi\mathbb S, \lVert v \rVert=1$
as we can then evoke \Cref{eq:lem:cea}. 
We denote the eigenspace of $E_0$ by $V_0 = \operatorname{span}\{\psi_0\}$. 
We can write any normed $v\in \mathcal H$ as $v = \cos(\theta)\psi_0+\sin(\theta) w$ orthogonally, where $\theta = \sphericalangle(v,\psi_0) = \langle v, \psi_0\rangle$ denotes the angle between $v$ and $\psi_0$ and $w$ depends on $v$. 
We obtain 
\begin{equation*}
    E(v) = \cos(\theta)^2 E_0 + \sin(\theta)^2 E(w) \ge \cos(\theta)^2 E_0 + \sin(\theta)^2 E_1
\end{equation*}
as we assumed a unique ground state. 
Hence, $E(v)\le\frac{E_0+E_1}{2}$ implies $\cos(\theta)^2\ge \sin(\theta)^2 $ and thus $\theta\in[-\frac\pi4, \frac\pi4]$ and analogously for a strict inequality. 
In particular, we obtain $\sphericalangle(\psi, \psi_0)\in(-\frac\pi4, \frac\pi4)$. 
Further, for any $v\in \mathcal H$ with $E(v)\le\frac{E_0+E_1}{2}$ we obtain $\sphericalangle(v, \psi_0)\in[-\frac\pi4, \frac\pi4]$ and thus $\sphericalangle(\psi, v)\in(-\frac\pi2,\frac\pi2)$ implying $v\notin T_\psi\mathbb S$. 
\end{proof}

Note that in the case of SR such an argument was not necessary as we always worked with $L^2$ projections. The above also crucially relies on the fact that we work in the tangent spaces $T_\psi\mathbb S$, as clearly the result cannot be correct when replacing $T_\psi\mathbb S$ by $\operatorname{dom}(\hat H)$, unless $\hat H$ was already positive definite.

\begin{remark}[Inf-sup conditions]\label{remark:inf_sup}
    To strengthen the norms in  \eqref{eq:quasi_best_approx}, we can consider the map 
    \begin{equation*}
        T_\psi\colon T_\psi\mathbb S\cap \operatorname{dom}(\hat H) \to T_\psi^*\mathbb S,
        \quad
        v\mapsto \langle v,\hat H\cdot\rangle - E(\psi)\langle v,\cdot\rangle
    \end{equation*}
    and show that it is bijective in a neighborhood of the ground state and thus satisfies inf-sup conditions. If furthermore inf-sup conditions hold on the discrete space $\mathcal V_\theta$, we can strengthen the norm of the left-hand side in \eqref{eq:quasi_best_approx} to the graph norm on $\operatorname{dom}(\hat H)$, see \cite{xu2003some}.
\end{remark}

\begin{lemma}[Newton Direction Representation]\label{lem:newton_direction_representation}
    Let $\psi\in\mathbb S \cap \operatorname{dom}(\hat H)$ and assume that $E(\psi)\notin\sigma(\hat H)$. Then \begin{equation}
    \label{eq:newton_direction_representation}
        d
        =
        \psi 
        - 
        \frac{(\hat H - E(\psi)I)^{-1}\psi}{\langle (\hat H - E(\psi)I)^{-1}\psi, \psi \rangle}. 
    \end{equation}
    is a member of $T_\psi\mathbb S \cap \operatorname{dom}(\hat H)$ and solves~\eqref{eq:characterization-Newton-direction}. 
    Under the assumptions of \Cref{lem:cea-newtons-method}, the Newton direction is unique. 
\end{lemma}
\begin{proof}
    By assumption, $\hat H - E(\psi) I$ is invertible, which makes $d$ given by equation \eqref{eq:newton_direction_representation} well defined. It is straightforward to verify that $d$ satisfies the equation in the tangent space and is in fact a member of $T_\psi\mathbb S$.
\end{proof}

Substituting equation \eqref{eq:newton_direction_representation} into the Newton iteration in \eqref{eq:newton_sphere_1}\&\eqref{eq:newton_sphere_2} allows us to rewrite it in a single line
\begin{equation*}
    \psi_{k+1}
    =
    \frac{(\hat H - E(\psi_k)I)^{-1}\psi_k}{\|(\hat H - E(\psi_k)I)^{-1}\psi_k\|},
\end{equation*}
which is a well known iterative eigensolver known as \emph{Rayleigh quotient iteration}~\cite{trefethen2022numerical}. 
Hence, the Riemannian Newton method corresponds to Rayleigh quotient iteration and converges to an eigenstate corresponding to the eigenvalue closest in energy to the initial energy  $E(\psi_0)$ at a cubic rate in the finite-dimensional setting~\cite{trefethen2022numerical}.

\subsection{Globalized Riemannian Newton}\label{sec:inverse_iter}
In the previous section we saw, that the Riemannian Newton method corresponds to Rayleigh quotient iteration and converges to the closest eigenstate in energy. To guarantee global convergence to the ground state we need a globalization strategy. To this end, we consider a shifted Newton iteration and show that it corresponds to shifted inverse iteration. If a priori knowledge on the ground state energy is available, this can be leveraged for the optimal choice of shift parameter, compare also to the numerical experiments in \Cref{sec:computational-examples}.

\subsubsection*{Synopsis}
To control the spectrum of $D^2E$ and prevent negative eigenvalues, we employ a diagonal shift, for a discussion of this globalization strategy in the Euclidean case see \cite[Chapter 3.4]{nocedal1999numerical}. Precisely, given an initial iterate $\psi_0\in\mathbb S \cap \operatorname{dom}(\hat H)$ and shift values $\tau_k$, we consider the scheme
\begin{align}
    \label{eq:global_newton_sphere_1}
    d_k 
    &=
    -(D^2E(\psi_k)+\tau_kI)^{-1}[DE(\psi_k)],
    \\
    \label{eq:global_newton_sphere_2}
    \psi_{k+1}
    &=
    \frac{ \psi_k - d_k}{\| \psi_k - d_k\|},
\end{align}
Writing out the Newton equation explicitly yields
\begin{equation}
    \langle
        [
            \hat H + (\tau_k - E(\psi_k))I
        ]d_k,
        v
    \rangle
    =
    \langle 
        \hat H \psi_k, v
    \rangle
    \quad 
    \text{for all }v\in T_{\psi_k}\mathbb S.
\end{equation}
Analogous to the case of Newton's method, we can rewrite the iterates as 
\begin{equation}
    \psi_{k+1}
    =
    \frac{(\hat H - (E(\psi_k)-\tau_k) I)^{-1}\psi_k}{\| (\hat H - (E(\psi_k)-\tau_k))^{-1}\psi_k \|},
\end{equation}
thereby offering an interpretation of the regularized Newton scheme~\eqref{eq:global_newton_sphere_1}\&\eqref{eq:global_newton_sphere_2} as shifted inverse iteration~\cite{trefethen2022numerical}.
Translating to parameter space, as in the case of Newton's method, this yields a preconditioned gradient method. Adding regularization $\varepsilon_k$, and introducing a step-size $\eta_k$ \eqref{eq:newton_sphere_2} we obtain
\begin{equation}\label{eq:discretized_global_newton_scheme}
    \theta_{k+1}
    =
    \theta_k
    -\frac{\eta_k}{2}
    \left(
    H(\theta_k) + (\tau_k - L(\theta_k))S(\theta_k) + \varepsilon_k I
    \right)^{-1}\nabla L(\theta_k).
\end{equation}
To obtain meaningful updates, the preconditioning matrix $H(\theta_k) +(\tau_k -L(\theta_k))S(\theta_k)$ in the formula \eqref{eq:discretized_global_newton_scheme} should be positive semi-definite. In fact, the ideal choice for the shift is $\tau_k = L(\theta_k) - E_0$, where $E_0$ is the ground state energy. In this case the globalized Newton scheme is identical to inverse iteration with the optimal shift $E_0$
\begin{equation*}
    \psi_{k+1}
    =
    \frac{(\hat H - E_0 I)^{-1}\psi_k}{\| (\hat H - E_0 I)^{-1}\psi_k \|}.
\end{equation*}
In practice, $E_0$ is not or only approximately known. We focus our analysis on the case where $\tau = \tau_k - L(\theta_k)$ is constant but not necessarily equal to the ground state energy.
We summarize our main insights as follows: 
\begin{itemize}
    \item Newton's method is locally attracting and converges to the closest eigenstate in energy. This does not translate directly to parameter space, as preconditioning matrices are indefinite and thus difficult to regularize. Consequently, convergence to the ground state relies on a good initial guess.
    \item Globalizing Newton's method by a diagonal shift corresponds to inverse iteration. This allows principled choices of shifts---especially if the ground state energy is approximately known. Choosing $\tau=\tau_k-L(\theta_k)$ fixed, the convergence rate is linear and the rate of the linear convergence improves the closer $\tau$ is to the ground state energy.
\end{itemize}

\subsubsection*{Mathematical Derivations}

To analyze the convergence of the shifted Newton method \eqref{eq:discretized_global_newton_scheme}, we use its connection to inverse iteration and assume that $\tau_k$ is chosen such that $\tau=\tau_k-L(\theta_k)$ is fixed.
\begin{theorem}[Convergence of Inverse Iteration]\label{theorem:convergence_shifted_inverse_iteration}
    Assume that the Hamiltonian $\hat H$ is a densely defined, self-adjoint operator on $\mathcal H$. Let $\tau \in \mathbb R$ be a shift that is not in $\sigma(\hat H)$ and assume that there exists a closest point $E_J$ to $\tau$ in the spectrum $\sigma(\hat H)$ and that $E_J$ is in the point spectrum. Assume that $E_J$ is isolated in the sense that 
    \begin{equation*}
        |\tau - E_J|
        <
        \inf_{E\in\sigma(\hat H)\setminus\{E_J\} }
        |\tau - E| = \zeta
    \end{equation*}
    We denote the eigenspace corresponding to $E_J$ by $V$ and assume that the initial function $\psi_0 \in \mathcal H$ has nonzero overlap with $V$, meaning that  
    the projection $P_{V}\psi_0$ of $\psi_0$ onto $V$ is non-zero and set $v_J\coloneqq \frac{P_{V}\psi_0}{\lVert P_{V}\psi_0 \rVert}\in V$. 
    Then the iterates 
    \[
        \psi_{k+1}
        =
        \frac{(\hat H - \tau I)^{-1}\psi_k}{\| (\hat H - \tau I)^{-1}\psi_k \|}
    \] 
    generated by the inverse iteration method with shift $\tau$ satisfy
    \begin{equation}
        \|\psi_k \pm v_J \| 
        =
        \mathcal O
        \left(
        \left|\frac{E_J - \tau}{\zeta}\right|^k
        \right),
        \quad
        |E(\psi_k) - E_J|
        =
        \mathcal O
        \left(
        \left|\frac{E_J - \tau}{\zeta}\right|^{2k}
        \right). 
    \end{equation}
\end{theorem}
\begin{proof}
    Denote the shifted inverse by $T_\tau = (\hat H - \tau I)^{-1}$. Our assumptions guarantee that the spectrum $\sigma(T_\tau)$ attains its maximum in $\mu_1 = (E_J - \tau)^{-1}$ and furthermore the assumption implies
    \begin{equation*}
        \zeta^{-1} = \sup \sigma(T_\tau)\setminus\{ \mu_1 \} =\mu_2 < \mu_1.
    \end{equation*}
    We denote the eigenspace of $\mu_1$ by $V$ and as $\hat H$ is a closed operator the eigenspace $V$ is closed so that we can consider the orthogonal decomposition of $\mathcal H = V + W$. Note that $T_\tau$ leaves these spaces invariant. Now, we write the initial condition as $\psi_0 = v_0 + w_0$ and consider the $k$-th iterate which can be written as
    \begin{equation*}
        \psi_k
        =
        c_k\mu_1^k\left[
            v_0 + \mu_1^{-k}T_\tau^kw_0
        \right]
        =
        c_k\mu_1^k\left[
            v_0 + \varepsilon_k
        \right],
    \end{equation*}
    where $c_k$ is chosen such that $\|\psi_k\|=1$. We now estimate
    \begin{equation*}
        \left \| \psi_k - \operatorname{sign}(\mu_1)^k\frac{v_0}{\|v_0\|} \right \|
        \leq
        \left|
            c_k\mu_1^k\|v_0\| - \operatorname{sign}(\mu_1)^k
        \right|
        +
        c_k\mu_1^k \|\varepsilon_k\|
    \end{equation*}
    and it remains to show that both summands above are of the order
    \begin{equation*}
        \mathcal O\left( \left|\frac{\mu_2}{\mu_1}\right|^k \right)
        =
        \mathcal O\left( \left| \frac{E_J - \tau}{\zeta} \right|^k \right).
    \end{equation*}
    Let us begin with $\varepsilon_k$, which we can estimate as follows
    \begin{equation*}
        \|\varepsilon_k\|
        \leq 
        \mu_1^{-k}\|T^k_\tau w_0\|
        \leq 
        \|w_0\|\left|\frac{\mu_2}{\mu_1}\right|^k
        \leq 
        \|\psi_0\|\left|\frac{\mu_2}{\mu_1}\right|^k,
    \end{equation*}
    where we used that the restriction of $T_\tau$ to $W$ is self-adjoint and hence its operator norm coincides with the supremum $\mu_2$ of its spectrum. Further, as $\|\psi_k\|=1$ it follows $|c_k\mu_1^k| = \|v_0 + \varepsilon_k\|^{-1} \leq (\|v_0\| - \|\varepsilon_k\|)^{-1} \leq \frac12\|v_0\|^{-1}$ for all $k \geq K_0$, where $K_0=K_0(\|\psi_0\|, |\frac{\mu_2}{\mu_1}|, \|v_0\|)$ is chosen such that $\|\varepsilon_k\|\leq\frac12\|v_0\|$. We estimate for $k>K_0$
    \begin{equation*}
        \begin{split}
            \left|
            c_k\mu_1^k\|v_0\| - \operatorname{sign}(\mu_1)^k
        \right|
        & =
        \left|
            \operatorname{sign}(\mu_1)^k\mu_1^k c_k\|v_0\| - 1
        \right|
        \\ &=
        \left|
            |\mu_1^kc_k|\|v_0\| - 1
        \right|
        \\
        &=
        \left|
            \frac{\|v_0\|}{\|v_0 + \varepsilon_k\|} -1
        \right|
        \\
        &=
        \vert c_k\mu_1^k \vert
         \cdot\left| \| v_0\| - \|v_0 + \varepsilon_k\| \right|
        \\
        &\leq
        |c_k\mu_1^k|\|\varepsilon_k\|
        \\&\leq
        2\frac{\|\psi_0\|}{\|v_0\|}\left| \frac{\mu_1}{\mu_2}\right|^k.
        \end{split}
    \end{equation*}
    Combining the above estimates yields 
    \begin{equation*}
        \left\|
            \psi_k - \operatorname{sign}(\mu_1)^k\frac{v_0}{\|v_0\|}
        \right\| 
        \leq  
        4\frac{\|\psi_0\|}{\|v_0\|}\left| \frac{\mu_1}{\mu_2}\right|^k,
    \end{equation*}
    which concludes the proof.
\end{proof}

\begin{remark}[Closing spectral gap]
    Assume that we choose a shift $\tau<E_0$, then we obtain linear convergence with the constant $
    (E_1-\tau)^{-1}(E_0-\tau)$. 
    When we choose the shift depending on the spectral gap $\tau=
    (1-\gamma)^{-1}(E_1-\gamma E_0)$ we obtain a convergence rate of $\gamma\in(0,1)$, which is independent of the spectral gap. 
\end{remark}

\paragraph{Derivation from inverse iteration}
We conclude with a self-contained derivation of the scheme \eqref{eq:discretized_global_newton_scheme} as a Galerkin-projected shifted inverse iteration. 
This is not strictly necessary for the theoretical analysis, yet we believe it underlines the usefulness of the proposed viewpoint of Galerkin-discretized algorithms. 
To this end, recall that inverse iteration with a fixed shift value $\tau\in\mathbb R$ 
and initial condition $\psi_0\in\operatorname{dom}(\hat H)$ with $\|\psi_0\|=1$ is given by 
\begin{equation*}
    \psi_{k+1} = \frac{(\hat H -\tau I)^{-1}\psi_k}{\|(\hat H -\tau I)^{-1}\psi_k\|}.
\end{equation*}
We cannot directly discretize the above expression, hence we rearrange to obtain the form $\psi_{k+1}= \psi_k - d_k$
\begin{equation*}
    \psi_{k+1}
    =
    \psi_k
    -
    \left[
        \psi_k
        -
        \frac{(\hat H - \tau I)^{-1}\psi_k}{\|(\hat H - \tau I)^{-1}\psi_k\|}
    \right]
    =
    \psi_k
    -
    d_k.
\end{equation*}
Then $d_k$ can be rewritten as
\begin{equation*}
    d_k
    =
    \frac{(\hat H - \tau I)^{-1}
    \left[
        \|(\hat H - \tau I)^{-1}\psi_k\| (\hat H - \tau I)\psi_k - \psi_k 
    \right]
    }{\|(\hat H - \tau I)^{-1}\psi_k\|}.
\end{equation*}
To transfer the scheme to parameter space, we assume that $\psi_k = \psi_{\theta_k}$ and we will drop the subscript $k$ in the following. Then, the Galerkin discretization of the map $\hat H - \tau I:\operatorname{dom}(\hat H)\to \mathcal H^*$ with $\psi\mapsto \langle \hat H\psi - \tau\psi,\cdot\rangle$ in the space $\mathcal V_\theta$ corresponds to the matrix $H(\theta) - \tau S(\theta)$. For the right-hand side, its action on $\partial_{\theta_i}\hat\psi_\theta$ is given by  
\begin{equation*}
    \|(\hat H - \tau I)^{-1}\psi_\theta\|
    \langle 
        \hat H\psi_\theta - \tau\psi_\theta, 
        \partial_{\theta_i}\hat\psi_\theta
    \rangle
    -
    \langle  
        \psi_\theta, 
        \partial_{\theta_i}\hat\psi_\theta 
    \rangle
\end{equation*}
Using that $\psi_\theta\perp \partial_{\theta_i}\hat\psi_\theta$, the above expression equals
\begin{equation*}
    \frac12\|(\hat H - \tau I)^{-1}\psi_\theta\| \|\psi_\theta\|\nabla L(\theta)_i.
\end{equation*}
Hence, the Galerkin discretization of $d$ leads to the system
\begin{equation*}
    [H(\theta) - \tau S(\theta)]\xi 
    =
    -\frac{\|\psi_\theta\|}2\nabla L(\theta).
\end{equation*}
Note that for $\|\psi_\theta\| = 1$ and $\tau = L(\theta)$, this is exactly the discretized version of the Riemannian Newton method introduced in the previous section. For different values of $\tau$, it corresponds to shifted inverse iteration.

\section{Related Literature}
Stochastic Reconfiguration is by far the most popular method for variational wavefunction optimization in VMC. It has been proposed in \cite{sorella1998green, sorella2001generalized} and has since been widely used, both for classic wavefunction ansätze \cite{becca2017quantum, cuzzocrea2020variational}, tensor networks~\cite{liu2025accurate, emonts2020variational, vieijra2021many, wu2025algorithms} and neural quantum states~\cite{carleo2017solving, pescia2024message, vicentini2022netket}. 
It has repeatedly been demonstrated to outperform the established first-order optimization methods, such as Adam and SGD. 
While the classic literature favors large sample sizes and moderate parameter space dimensions, this has drastically changed with the advent of neural wavefunctions, which feature high-dimensional parameter spaces. In this case, the cubic scaling of SR is either handled through Kronecker-factored approximations to the overlap matrix, typically referred to as K-FAC~\cite{martens2015optimizing, pfau2020ab}, or for moderate sample sizes by using the Woodbury identity~\cite{chen2024empowering, rende2024simple}, compare also to \Cref{sec:preliminaries}. In this manuscript, we employ the Woodbury identity rather than K-FAC, although the latter is a viable option, especially building on recent work for Kronecker-factored approximations of preconditioning matrices that involve PDE operators~\cite{dangel2024kronecker}.

In the literature, stochastic reconfiguration is often derived via a time-dependent variational principle~\cite{yuan2019theory} applied to the imaginary time evolution of the Schrödinger equation. 
This viewpoint is similar to our derivation. 
Note, however, that the viewpoint of time-dependent variational principles does not allow for the translation of generic optimization algorithms in the way that we put forth in \Cref{sec:functional-framework}. 

While stochastic reconfiguration is presently the standard method for VMC optimization, there have been several studies exploring different options. One method, called the \emph{linear method} \cite{toulouse2007optimization, nightingale2001optimization, motta2024quantum, umrigar2007alleviation}, explicitly minimizes the Rayleigh quotient in the affine space $\hat\psi_\theta + \mathcal V_\theta$ at every optimization iteration. 
Other works consider the full parameter space Newton method \cite{toulouse2007optimization}. 
The recently proposed Wasserstein quantum Monte Carlo method \cite{neklyudov2023wasserstein} can be derived from our functional viewpoint---as a Galerkin discretization of Wasserstein gradient descent. 
Finally, the Rayleigh-Gauss Newton method \cite{webber2022rayleigh} can be understood within our functional framework as a variant of the globalized Riemannian Newton method. 
In the exposition of the Rayleigh Gauss-Newton method, the matrix $H(\theta)-L(\theta)S(\theta)$ is motivated as an approximation to the full parameter Hessian close to an eigenstate. 
Our findings stress a different perspective---the matrix $H(\theta)-L(\theta)S(\theta)$ is the \emph{Riemannian Hessian} in a suitable basis and can be understood as a projected version of Rayleigh quotient iteration. 
Moreover, it is questionable whether the full parameter Hessian is a desirable preconditioner, see \Cref{remark:newton_for_rq_ill_suited}.
For a complete picture, we review the approach of \cite{webber2022rayleigh}. To this end, for given parameters $\theta\in\Theta$ and a perturbation $\delta\theta\in\Theta$, the Taylor expansion up to second-order of $L$ is
\begin{equation}
    L(\theta + \delta\theta)
    \approx
    L(\theta)
    +
    \nabla L(\theta)\delta\theta
    +
    \delta\theta^\top[H(\theta) - L(\theta)S(\theta) + D(\theta)]\delta\theta,
\end{equation}
where $S(\theta)$ and $H(\theta)$ denote the overlap and Hamiltonian matrix and the matrix $D(\theta)$ contains second derivatives of the normalized wavefunction. It is given by
\begin{equation*}
    D(\theta)_{ij}
    =
    \left\langle
        (\hat H - L(\theta)I)\left( \frac{\psi_\theta}{\|\psi_\theta\|} \right),
        \partial_{\theta_i}\partial_{\theta_j}\left( \frac{\psi_\theta}{\|\psi_\theta\|} \right)
    \right\rangle.
\end{equation*} 
The realization of \cite{webber2022rayleigh} is that at an eigenstate of the Hamiltonian $\hat H$, the term $D(\theta)$ vanishes. 
Together with the high computational cost of second derivatives, the authors motivate the use of $H(\theta) - L(\theta)S(\theta)$ as an approximation of the full Hessian. Furthermore, they propose a regularization of $H(\theta) - L(\theta)S(\theta)$ by adding a multiple of the overlap matrix $S(\theta)$. 
In the light of \Cref{sec:inverse_iter}, the regularization term corresponds to a globalization strategy in the Riemannian Newton method. 
Note that our functional in \Cref{sec:inverse_iter} allows a principled choice of this regularization parameter in the case when an estimate on the ground state energy is available, thus eliminating an ad-hoc choice of this hyperparameter. 
We conclude this section with a well-known fact on the suitability of the full parameter Hessian of $L$.

\begin{remark}\label{remark:newton_for_rq_ill_suited}
    Assume a linearly parametrized wavefunction or the functional, non-parametrized setting. Then, Newton's method applied to the Rayleigh Quotient does only rescale iterates and is therefore not suitable for wavefunction optimization. In fact, using the formulas \eqref{eq:first-derivative-RQ} and \eqref{eq:second-derivative-RQ}, one verifies that the current iterate $\psi_k$ satisfies
    \begin{equation*}
        D^2E(\psi_k)(\psi_k,\cdot)
        =
        -
        \frac{2}{\|\psi_k\|^2}
        \left[
            \langle \hat H\psi_k,\cdot\rangle - E(\psi_k)\langle \psi_k, \cdot \rangle
        \right]
        =
        -DE(\psi_k),
    \end{equation*}
    hence Newton's method with the full Hessian merely rescales iterates $\psi_{k+1} = 2\psi_k$.
\end{remark}

\section{Computational Examples}\label{sec:computational-examples}
This section contains numerical examples illustrating how the theoretical insight into the function space algorithms can be used for the understanding and design of schemes in parameter space. 
In the numerical experiments, we mainly compare stochastic reconfiguration (SR) to the projected inverse iteration (PII) scheme derived from the globalized Newton method \eqref{eq:discretized_global_newton_scheme} (or equivalently derived from the functional version of inverse iteration).  We recall the PII scheme in parameter space
\begin{align}\label{eq:II}
    \begin{split}
    \theta_{k+1}
    &=
    \theta_k
    -
    \eta_k
    (H_k - \tau_kS_k + \varepsilon_k I)^{-1}\nabla L(\theta_k)
    \\
    &=
    \theta_k 
    -
    \eta_k
    (O_k^\top A_k - \tau_k O_k^\top O_k + \varepsilon_k I)^{-1}O_k^\top r_k 
    \end{split}
\end{align}
Here, $H_k$ is the estimator of the Hamiltonian matrix introduced in \eqref{eq:H_matrix_estimator}, $S_k$ is the estimator of the overlap matrix \eqref{eq:overlap_matrix}, $\eta_k$ is the learning rate, $\varepsilon_k$ is a regularization parameter referred to as the diagonal shift, and $\nabla L(\theta_k)$ is the estimated energy gradient. The parameter $\tau_k < 0$ corresponds to the spectral shift of the Hamiltonian, we typically choose it in dependence of the ground state energy $\tau_k = \alpha E_0$, with $\alpha\geq1$. In the overparametrized setting, where the number of parameters is larger than the number of samples, i.e., $M \ll P$, we use the update formula
\begin{equation}\label{eq:minPII}
    \theta_{k+1} 
    =
    \theta_k 
    -
    \eta_k
    O_k^\top(A_k O_k^\top - \tau_k O_k O_k^\top + \varepsilon_k I)^{-1} r_k.
\end{equation}
In resemblance of minimal norm stochastic reconfiguration (minSR) \cite{chen2024empowering, rende2024simple}, we refer to this update scheme as \emph{minPII}. We stress that it is---up to floating point arithmetic---identical to \eqref{eq:II}.

\paragraph{Experimental setup}
The simulations presented in \Cref{subsec:TFIM-fullsum} and \Cref{subsec:TFIM-Heisenberg} have been produced using an implementation based on NetKet~\cite{vicentini2022netket}.  

\subsection{Transverse field Ising $4\times 4$ with full Hilbert sum}\label{subsec:TFIM-fullsum}
We begin with a transverse field Ising model (TFIM) 
\begin{equation}
    \hat{H} = - J \sum_{\langle i, j\rangle} \hat{Z}_i \hat{Z}_j - h \sum_{i=1}^{N} \hat{X}_i
\end{equation}
where $\hat{Z}_i$ and $\hat{X}_i$ are the Pauli matrices applied to site $i$. We use periodic boundary conditions. Our wave function is represented in the $\hat{Z}$ eigenbasis.
We consider a $4\times 4$ lattice, which is small enough to allow the computation of expectation values by summing over the full Hilbert space, bypassing Monte Carlo sampling. We refer to this approach as the \emph{full sum} variational optimization. 
This enables a clear inspection of the respective algorithms independent of sampling noise. 
For PII, we test the shift values $\tau_k = E_0$ and $\tau_k = 1.4 \times E_0$. The learning rates $\eta$ and diagonal shifts $\varepsilon$ for all methods are manually tuned for fastest convergence. We report the relative energy errors over the first $100$ iterations in~\Cref{fig:tfim4x4}. 

\begin{figure}[h]
    \centering
    \includegraphics[width=\linewidth]{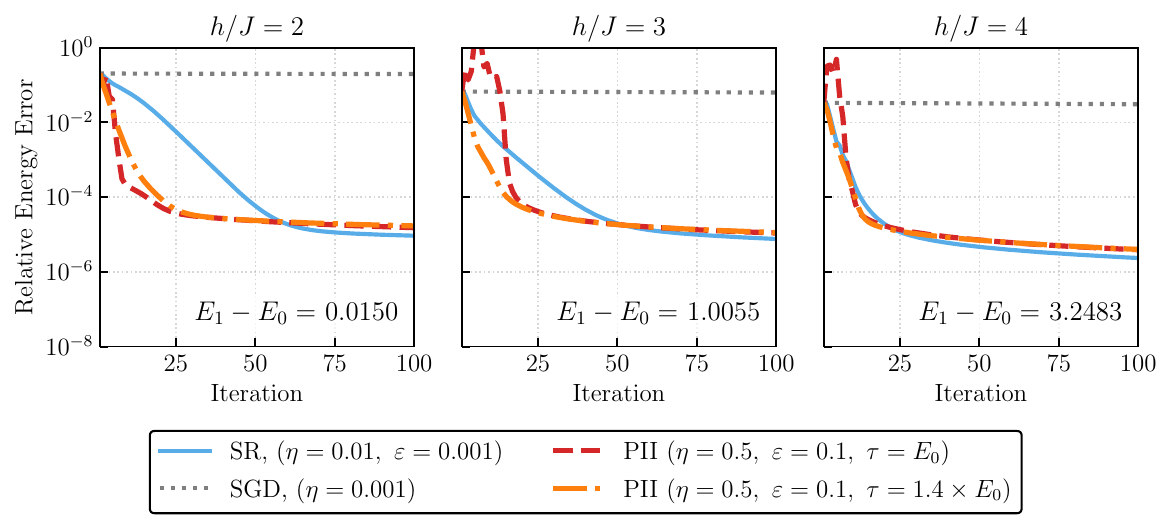}
    \caption{Full sum TFIM $4\times 4$. Performance of stochastic reconfiguration (SR), projected inverse iteration (PII) corresponding to a globalized Riemannian Newton method and stochastic gradient descent (SGD) for a full-sum $4\times 4$ transverse field Ising model. The learning rate is denoted by $\eta$, the regularization parameter is $\varepsilon$, and $\tau$ is the shift parameter in inverse iteration.}
    \label{fig:tfim4x4}
\end{figure}

\subsubsection*{Robustness of PII to narrow spectral gaps} 
We observe that for all problems, SGD stagnates, where the other methods initially decrease the energy quickly and slow down after reaching a relative energy error of around $10^{-5}$.
While the convergence speed of SR differs notably for the three cases $h/J=2,3,4$, this is less pronounced for PII. In fact, PII provides the fastest optimization across the three test cases and the optimization trajectory changes little when varying $h/J$. 
This is in agreement with our theory: The convergence of SR depends on the spectral gap $E_1 - E_0$, see \Cref{thm:convergence_speed}, whereas this dependence can be mitigated for PII by choosing a shift $\tau$ close to $E_0$, see \Cref{theorem:convergence_shifted_inverse_iteration}. 
In this example, we can compute the spectral gap exactly, resorting to a Lanczos iteration, and observe that a deteriorating convergence speed of SR directly relates to a closing spectral gap, compare to the case $h/J=2$, where SR is slowest and the reported spectral gaps in \Cref{fig:tfim4x4}. To be precise the energy reduction factor per iteration step for PII with unit learning rate is 
\begin{equation*}
    r_{\text{PII}} = \left| \frac{E_0 - \tau}{E_1 - \tau}  \right|^2.
\end{equation*}
Choosing $\tau$ close enough to $E_0$ hence prevents a slow-down of the iteration, even for small spectral gaps. In contrast, the energy reduction factor for SR with learning rate $\eta$ shrinks exponentially as the gap $E_1 - E_0$ closes as can be seen from the formula
\begin{equation*}
    r_{\text{SR}} = \left[ \exp(-2|E_1 - E_0|\eta) \right]^2.
\end{equation*}

\subsubsection*{Sensitivity to Sampling} 
For the theoretically optimal shift value of $\tau=E_0$, we frequently observe instable behavior or optimization failure for PII. 
\Cref{fig:tfim4x4_sample_sensitivity} illustrates this instability and failure to optimize are even more pronounced when MCMC sampling is used. 
In this case, only large sample sizes yield a convergent method. However, using a theoretically suboptimal shift value, in this case $\tau=1.2 E_0 <E_0$ suffices to stabilize the method, also in the small sample regime. 
We attribute this stabilization to the improved spectral properties of the preconditioner: The spectrum of $H - \tau E_0 S$ is moved to the right 
as $\tau$ increases, which heuristically makes it easier to retain the positive semi-definiteness when estimated. 
The definiteness, in turn, guarantees a meaningful regularization of the preconditioner when a diagonal shift is added and secures a descent direction in the optimization.
\begin{figure}[h]
    \centering
    \includegraphics[width=0.95\linewidth]{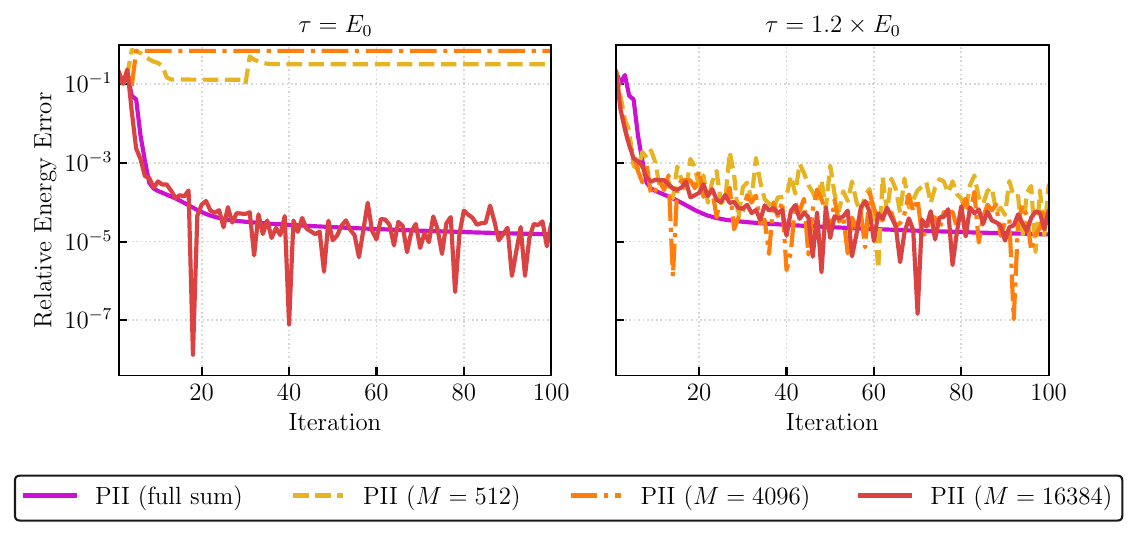}
    \caption{TFIM $4\times 4$, $h/J=2$. Illustration of the effect of sample size $M$ and shift parameter $\tau = \alpha\times E_0$. All optimization trajectories use the regularization parameter $\varepsilon=0.1$ and the learning rate $\eta=0.5$.}
    \label{fig:tfim4x4_sample_sensitivity}
\end{figure}

\subsubsection*{Large learning rates of PII} 
\begin{wrapfigure}{r}{0.53\textwidth}
    \centering
    \includegraphics[width=\linewidth]{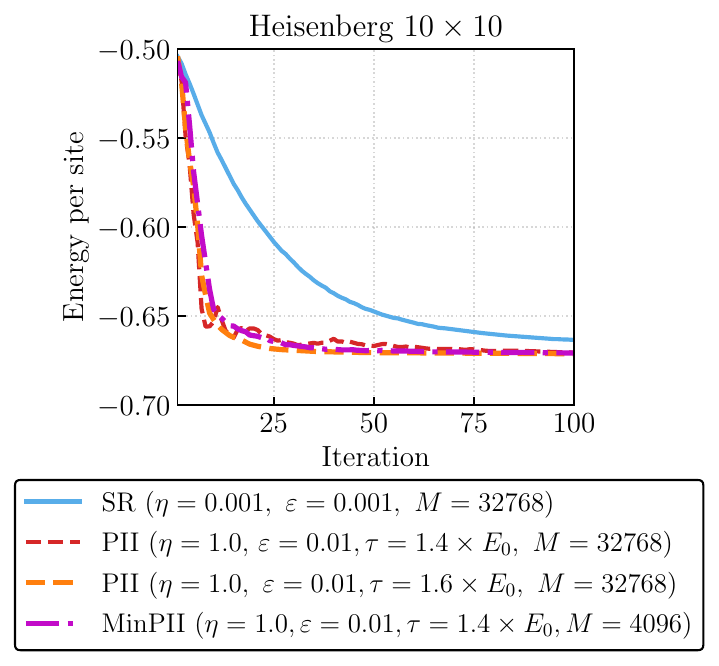}
    \caption{Heisenberg $10\times 10$: Performance of SR and PII using a Vision Transformer with 18,140 trainable parameters and varying hyperparameter settings.}
    \label{fig:heisenberg10x10_ViT_log}
\end{wrapfigure}
We emphasize the drastic differences in learning rates for SR and PII, with PII being able to use high learning rates, often around $0.5$. 
This difference can be explained on the functional level: SR derives from an explicit Euler discretization of the Riemannian $L^2$ flow on $E$ and hence requires small enough time steps in order to remain convergent. 
PII, on the other hand, is derived from inverse iteration on the functional level, where the natural learning rate is 1.0 (or 0.5 depending on the scaling). 
Both methods are Galerkin-projected versions of functional algorithms, hence, transferring learning rates from function space to the ansatz is not guaranteed. 
However, we observe empirically that PII works well with large learning rates. 
We emphasize again that we manually tuned all learning rates for maximal performance. 

\subsection{Transverse field Ising and Heisenberg $10 \times 10$}\label{subsec:TFIM-Heisenberg}
We now investigate the behavior of SR and PII for the same TFIM model on a $10\times 10$ lattice. Additionally, we consider a Heisenberg model 
\begin{align}
    \hat{H} &= \sum_{\langle i, j\rangle} \hat{X}_i \hat{X}_j + \hat{Y}_i \hat{Y}_j + \hat{Z}_i \hat{Z}_j
\end{align}
on the same $10\times 10$ lattice. Due to the exceedingly high dimension $2^{10\cdot10}\approx 10^{30}$ of the Hilbert space, exact computations are no longer feasible, and we resort to Monte Carlo estimators. 
Again, learning rates $\eta$ and diagonal shifts $\varepsilon$ are tuned for maximal performance. As ansätze, we consider both an RBM~\cite{carleo2017solving} and a vision transformer (ViT)~\cite{viteritti2023transformer} with 40400 trainable parameters for the RBM and 18140 trainable parameters for the ViT. For the latter, we take patches of size $2$, embedding dimension $20$, $2$ heads per layer, and $4$ layers.

We compare the performance of the methods for a large batch size of $M=32768$ and a modest one $M=4096$. 
In the case of the RBM, we use the formulation of PII that requires the solution of a linear system in sample size, which we call minPII, see \eqref{eq:minPII}.
\begin{figure}[h]
    \centering
    \includegraphics[width=\linewidth]{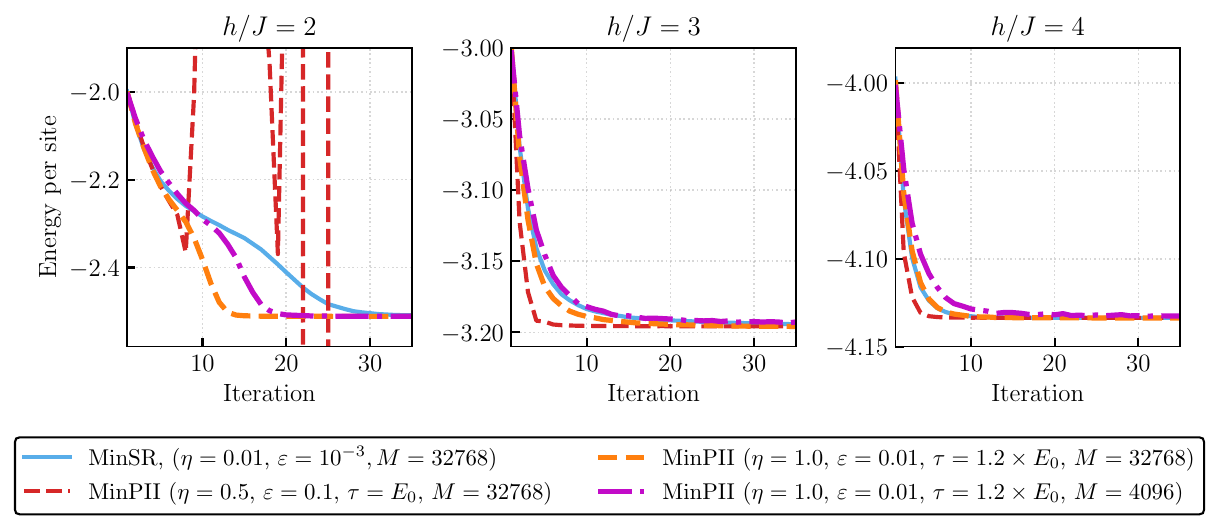}
    \caption{TFIM $10\times 10$: Shown are the results for a restricted Boltzmann Machine with 40400 trainable parameters, sample size 36768.}
    \label{fig:tfim10x10_rbm}
\end{figure}
In \Cref{fig:tfim10x10_rbm} we see that the findings of the small-scale TFIM model translate both to the $10\times10$ TFIM model and the Heisenberg model. For TFIM, the convergence difference of minPII and minSR is most pronounced in the case $ h/J=2$, where minPII is significantly more efficient; minPII works well in the large learning rate regime, and enlarging the shift value $\tau$ makes the method stable. In the case of the ViT, the difference between SR and PII is even more pronounced, much in favor of PII, a finding that we also observe for the Heisenberg model. In our simulations, we deliberately restrict to stoquastic Hamiltonians to ensure real and positive-valued parameterizations. A comprehensive benchmark and extension to Hamiltonians with non-trivial sign structures is left for future work.

\begin{figure}[h]
    \centering
    \includegraphics[width=\linewidth]{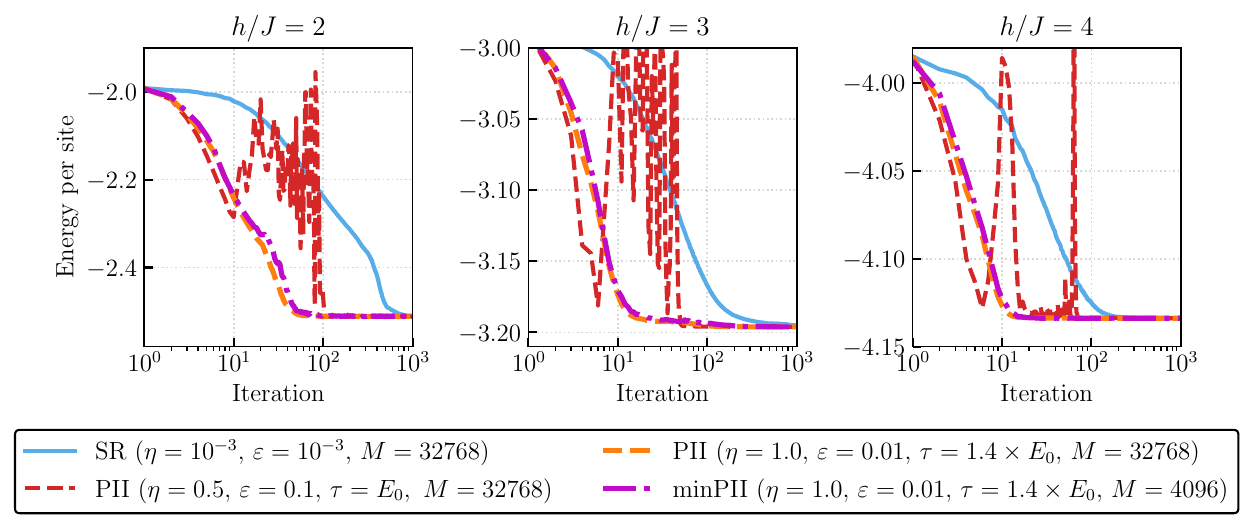}
    \caption{TFIM $10\times 10$: Vision Transformer with 18140 trainable parameters, sample size 36768 and 4096 respectively. }
    \label{fig:tfim10x10_ViT}
\end{figure}

\section{Conclusion and Outlook}
This manuscript presents a geometric framework for translating non-parametric function space algorithms into practical optimization schemes for variational quantum Monte Carlo. Our approach leverages Galerkin projections onto the tangent space of the variational ansatz, ensuring that the essential characteristics of the functional dynamics are preserved in the discretized, parameter-space algorithms. This unified perspective allows us to re-interpret existing methods, such as stochastic reconfiguration (SR), and derive novel, geometrically principled algorithms, notably the projected inverse iteration.

We demonstrate that stochastic reconfiguration can be understood as a Galerkin-projected Riemannian $L^2$ gradient descent on the sphere. Our analysis provides insights into its linear convergence rate, explicitly showing its dependence on the spectral gap of the Hamiltonian. In contrast, the Riemannian Newton method corresponds to the classical Rayleigh quotient iteration on the functional level, possessing a locally cubic convergence rate, at least in the finite-dimensional case. To address its local convergence, we propose a globalized Riemannian Newton method by introducing a diagonal shift. This method, when Galerkin-projected, yields our novel projected inverse iteration scheme, which directly corresponds to shifted inverse iteration in function space. Numerical experiments on various quantum spin systems (Transverse Field Ising Model, Heisenberg model) using neural network wavefunctions (Restricted Boltzmann Machines, Vision Transformers) confirm our theoretical findings. We observe that: 
\begin{itemize}
    \item PII consistently outperforms SR, especially in scenarios with narrow spectral gaps, where SR's convergence significantly deteriorates. This resilience stems from PII's connection to inverse iteration, which can mitigate the impact of small spectral gaps through a judicious choice of shift parameter.

    \item PII works well with large learning rates, aligning with its functional counterpart (inverse iteration) which naturally uses unit step size. This is in contrast with gradient descent methods that typically require smaller step sizes.

    \item The ability of PII to incorporate a priori approximate knowledge of the ground-state energy as a shift parameter provides a principled way to select this hyperparameter. While theoretically optimal shifts can sometimes be sensitive to sampling noise, slightly perturbed shifts empirically stabilize the method without significant loss of performance.
\end{itemize}

This work highlights the power of a geometric, function-space perspective in designing and understanding optimization algorithms for VMC. We believe this framework opens up many exciting avenues for future research. This includes exploring more advanced functional optimization tools like trust-region methods for globalization, gradient acceleration schemes, or leveraging conjugate gradient eigensolvers within the functional framework. The insights gained from the function space can continue to inform the design of more efficient and robust algorithms for challenging quantum many-body problems.

\appendix

\section*{Acknowledgments}
MZ acknowledges support from an ETH Postdoctoral Fellowship for the project “Reliable, Efficient, and Scalable Methods for Scientific Machine Learning”. This work was supported by a grant from the Swiss National Supercomputing Centre (CSCS) under project ID lp20 on Alps.

\bibliographystyle{siamplain}
\bibliography{references}

\end{document}

%% file: ex_shared.tex
\usepackage{todonotes}
\usepackage{lipsum}
\usepackage{amsfonts}
\usepackage{graphicx}
\usepackage{epstopdf}
\usepackage{algorithm}
\usepackage{algpseudocode}
\ifpdf
  \DeclareGraphicsExtensions{.eps,.pdf,.png,.jpg}
\else
  \DeclareGraphicsExtensions{.eps}
\fi

\newsiamremark{remark}{Remark}
\newsiamremark{hypothesis}{Hypothesis}
\crefname{hypothesis}{Hypothesis}{Hypotheses}
\newsiamthm{claim}{Claim}

\headers{Functional Neural Wavefunction Optimization}{Armegioiu, Carrasquilla, Mishra, Müller, Nys, Zeinhofer, Zhang}

\title{Functional Neural Wavefunction Optimization
}

\author{Victor Armegioiu\thanks{ETH Zürich, D-MATH  
  (\email{victor.armegioiu@sam.math.ethz.ch}).}
\and Juan Carrasquilla\thanks{ETH Zürich, D-PHYS
  (\email{jcarrasquill@ethz.ch}).}
\and Siddhartha Mishra\thanks{ETH Zürich, D-MATH
  (\email{siddhartha.mishra@sam.math.ethz.ch}).}
\and Johannes Müller\thanks{TU Berlin
    (\email{johannes.christoph.mueller@tu-berlin.de}).}
\and Jannes Nys\thanks{ETH Zürich, D-PHYS
    (\email{jannys@ethz.ch}).}
\and Marius Zeinhofer\thanks{ETH Zürich, D-MATH
    (\email{marius.zeinhofer@sam.math.ethz.ch}).}
\and Hang Zhang\thanks{ETH Zürich, D-PHYS
    (\email{hangzhang@ethz.ch}).}
}

\usepackage{amsopn}

\makeatletter
\newcommand*{\addFileDependency}[1]{
  \typeout{(#1)}
  \@addtofilelist{#1}
  \IfFileExists{#1}{}{\typeout{No file #1.}}
}
\makeatother

\newcommand*{\myexternaldocument}[1]{%
    \externaldocument{#1}%
    \addFileDependency{#1.tex}%
    \addFileDependency{#1.aux}%
}